\documentclass[final,1p,times]{elsarticle}

\usepackage{amsmath}
\usepackage{amsthm}
\usepackage{breakurl}
\usepackage[breaklinks]{hyperref}
\usepackage{amssymb}
\usepackage{lipsum}
\usepackage{listings}
\usepackage{gensymb}
\usepackage{fancyhdr}
\usepackage{graphicx}
\usepackage{url}
\usepackage{lineno}
\usepackage{xcolor}
\usepackage{color}
\newtheorem{thma}{Theorem}

\newtheorem{lemm}{Lemma}
\def\CC{{C\nolinebreak[4]\hspace{-.05em}\raisebox{.4ex}{\tiny\bf ++}}}

\journal{Computer Physics Communications}
\newcommand{\fracnot}[3]{\mathcal{F}_{#1}^{#2,#3}}

\begin{document}
\begin{frontmatter}



\title{Efficient GPU Thread Mapping on Embedded 2D Fractals}


\author[a]{Cristob\'al A. Navarro\corref{author}}
\author[a]{Felipe A. Quezada}
\author[b]{Nancy Hitschfeld}
\author[a]{Raimundo Vega}
\author[c]{Benjamin Bustos}

\cortext[author] {Corresponding author.\\\textit{E-mail address:} cristobal.navarro.g@gmail.com}
\address[a]{Instituto de Inform\'atica, Universidad Austral de Chile.}
\address[b]{Departamento de Ciencias de la Computaci\'on, Universidad de Chile.}
\address[c]{Millenium Institute Foundational Research on Data, Department of Computer Science, University of Chile.}

\begin{abstract}
This work proposes a new approach for mapping GPU threads onto a family of discrete embedded 2D fractals. A block-space map $\lambda: \mathbb{Z}_{\mathbb{E}}^{2} \mapsto \mathbb{Z}_{\mathbb{F}}^{2}$ is proposed, from Euclidean parallel space $\mathbb{E}$ to embedded fractal space $\mathbb{F}$, that maps in $\mathcal{O}(\log_2 \log_2(n))$ time and uses no more than $\mathcal{O}(n^\mathbb{H})$ threads with $\mathbb{H}$ being the Hausdorff dimension of the fractal, making it parallel space efficient.  When compared to a bounding-box (BB) approach, $\lambda(\omega)$ offers a sub-exponential improvement in parallel space and a monotonically increasing speedup $n \ge n_0$. The Sierpinski gasket fractal is used as a particular case study and the experimental performance results show that $\lambda(\omega)$ reaches up to $9\times$ of speedup over the bounding-box approach. A tensor-core based implementation of $\lambda(\omega)$ is also proposed for modern GPUs, providing up to $\sim40\%$ of extra performance. The results obtained in this work show that doing efficient GPU thread mapping on fractal domains can significantly improve the performance of several applications that work with this type of geometry.   
\end{abstract}


\end{frontmatter}

\section{Introduction}
\label{sec:intro}
Fractals can be described as self-similar structures \cite{mandelbrot2004} where a \textit{similar}\footnote{Depending on which fractal, the term \textit{similar}
can refer to \textit{exactly similar} or \textit{quasi similar}.} geometrical pattern is found at all scales.  Several natural phenomena produce fractal
patterns that obey a self-similar structure \cite{Mandelbrot:98509}, such as plant and tree growth \cite{Oppenheimer:1986:RTD:15886.15892,
Palmer1988}, terrain formation \cite{MILNE198867, 4767591}, molecular dynamics \cite{rothemund2004}, snowflake crystallization \cite{He2008739}, blood vessels
\cite{PhysRevLett.90.118101}, morphological features of living organisms \cite{WeibelL361}, among many others, display a fractal design where
self-similarity is a relevant feature for modeling its geometrical structure. 
Computer applications related to these fields may choose to 
embed the 2D fractal into a discrete Euclidean domain which acts as a bounding-box in memory. 
Using an embedding space for a fractal helps in achieving an efficient simulation in terms of memory access patterns (\textit{i.e.}, high memory bandwidth), as data-parallel computations (for example, computation of nearest neighbors) can perform aligned memory accesses and exploit the spatial locality within the fractal structure. In other words, memory locations $(x\pm1,y\pm1)$ define a neighborhood in both the embedded space and the actual fractal as well (some of the elements of this neighborhood belong to the fractal domain). Figure \ref{fig:h-fractal-embedded} illustrates as an example the H-fractal embedded in a space of $n \times n$ with $n$ being its its side length. 
\begin{figure}[ht!]
\centering
\includegraphics[scale=0.20]{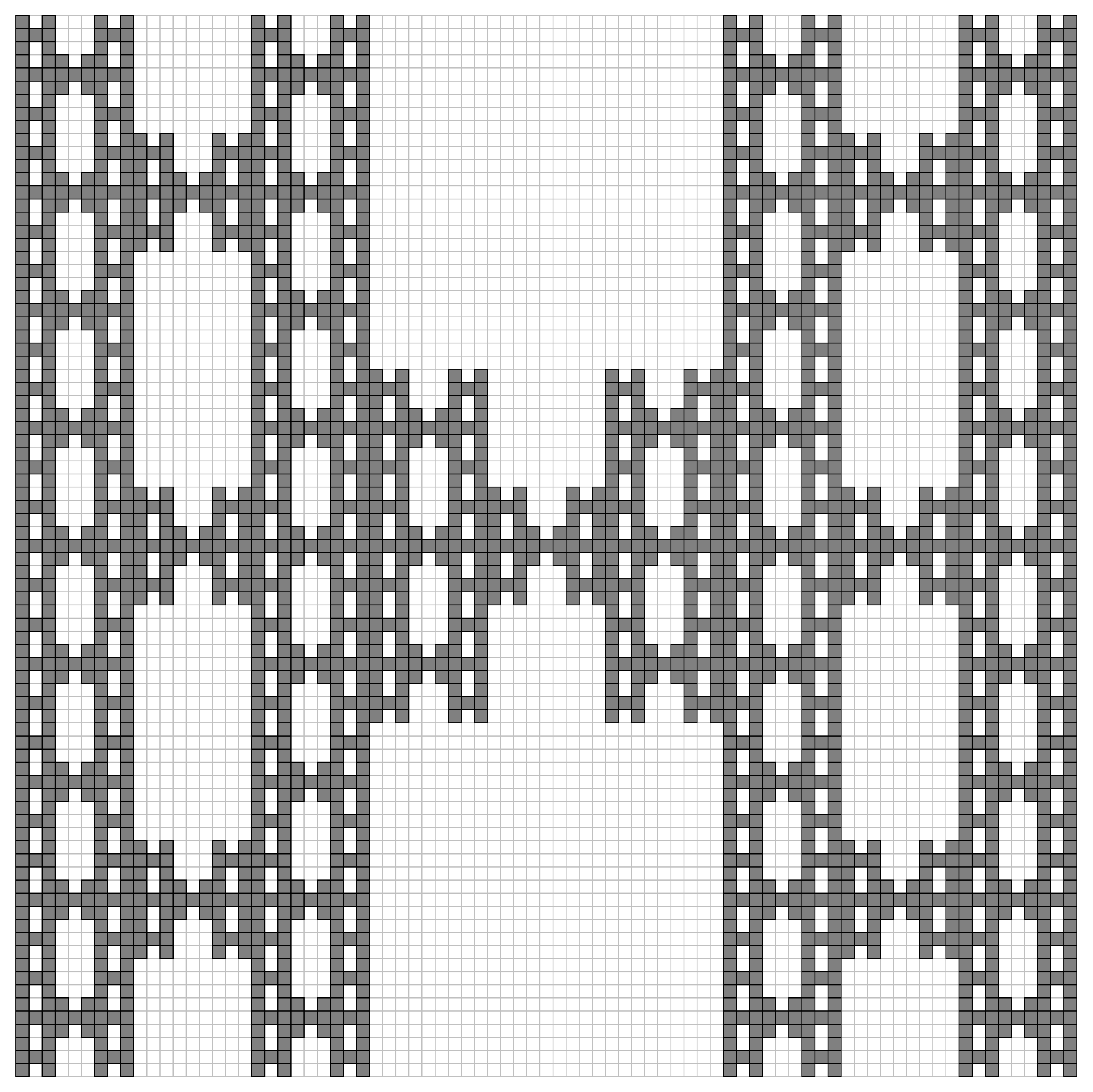}
\caption{A H-fractal embedded in a discrete $n \times n$ Euclidean space.}
\label{fig:h-fractal-embedded}
\end{figure}

Several of the processes performed on a fractal are combinations 
of \textit{map} and \textit{reduce} operations in different proportions. The \textit{map} component relates to the \textit{per-element} computations while the \textit{reduce} component relates to global results being computed from the parts of the problem (e.g. a global measure from all discrete elements of the fractal). One typical simulation pattern that is a combination of a full \textit{map} with multiple local \textit{reductions} is the tiled nearest-neighbors computation, which operates on all data elements of a fractal and for each one considers a small reduction from its neighborhood data. Such pattern is found in cellular automata transition functions, spin lattice Monte Carlo simulation steps and finite difference method (FDM) time-step computations, among others. Another computational pattern frequently used is the computation of a global measure of the fractal at a given state, which would involve a full \textit{reduction} of the values of all data elements of the fractal passed through a mathematical expression. This second pattern can be found when computing macroscopic measures from microscopic definitions, such as in Spin Lattice models or n-body simulations.
Eventually, when the fractal is large enough to the point of containing millions of data elements, a sequential computation can take an excessive amount of time for the practical requirements of a certain field, especially if there are real-time requirements. In these situations GPU computing becomes an attractive tool for accelerating these tasks \cite{navhitmat2014}. 

GPU computing has become an important tool for leveraging the performance of several compute demanding applications that contain data-parallel workloads \cite{navhitmat2014}. The two main motivations to use GPU Computing are the (1) high TFLOPS and memory bandwidth, which can be up to an order of magnitude faster than traditional CPU hardware, and (2) the high energy efficiency with respect to traditional CPU based systems. It is important to mention however that exploiting these two features efficiently requires a dedicated algorithm design and implementation as GPUs are more restricted than CPUs in terms of control logic (scheduling, branch prediction, prefetching) and memory access patterns (cache, global memory access). One aspect of GPUs that has been recently studied is achieving efficient GPU thread mapping, an optimization technique that can minimize the number of necessary threads when working with non-trivial data domains. In the GPU programming model, for every GPU computation there is a stage in the pipeline where threads are mapped from parallel work space to data space. A map, defined as $f: \mathbb{Z}^k \rightarrow \mathbb{Z}^m$, transforms each $k$-dimensional point $x=(x_1, x_2, ..., x_k)$ in parallel space $P^k$ into a unique $m$-dimensional point $f(x) = (y_1, y_2, ..., y_m)$ in data
space $D^m$. This work uses the notation introduced by previous GPU thread mapping works \cite{navarro2018competitiveness, 8291959, DBLP:conf/hpcc/NavarroH14}, which defines GPU parallel spaces as orthotopes\footnote{A $k$-orthotope is the generalization of the notion of a rectangle or box, for $k$ dimensions.}  $\Pi^k \in P^k$ in $k=1,2,3$ dimensions. In this work the orthotope of interest is the two-dimensional one, $\Pi^2$. A known way of mapping threads to any data-domain is to use the \textit{bounding-box} approach, that builds an orthotope $\Pi^2$ sufficiently large to cover the corresponding bounding-box of the data space and threads are mapped using the identity $f(\omega) = \omega$. Such a map is highly convenient and efficient for the class of problems where data space is also defined by an orthotope; such as vectors ($\Pi$), tables ($\Pi^2$), matrices ($\Pi^2$) and box-shaped volumes ($\Pi^3$). However, for a discrete embedded 2D fractal, this approach is no longer efficient in terms of parallel space as many threads would fall inside the embedded space but outside the fractal domain, introducing a performance penalty to the execution time when discarding these threads at run-time. For such cases, an efficient orthotope would be one with asymptotically the same number of threads as data elements of the fractal. Figure \ref{fig:vicsek-proposal-map} illustrates the unwanted (left to center) and wanted scenarios (right to center) for the case of the Vicsek fractal.


\begin{figure}[ht!]
\centering
\includegraphics[scale=0.045]{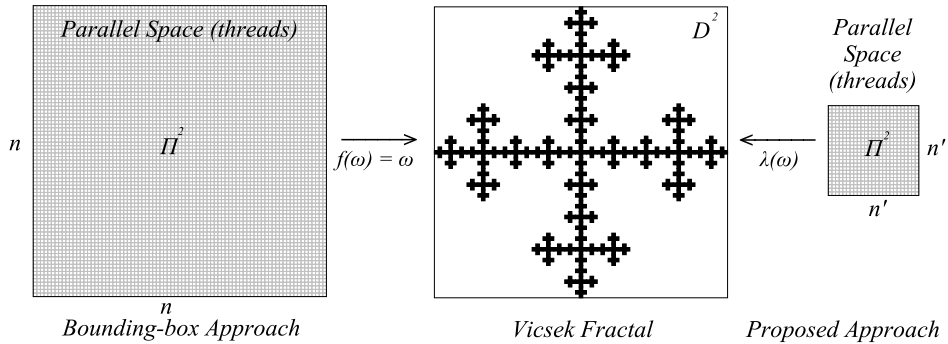}
\caption{GPU thread mapping for the Vicsek fractal. The left-to-center mapping illustrates the bounding-box approach and its cost in thread resources, while the right-to-center mapping shows the proposed approach, using significantly less thread resources, but a more elaborate map, namely $\lambda(\omega)$.}
\label{fig:vicsek-proposal-map}
\end{figure}

Two research questions arise from this GPU efficiency problem; the first question:
\textit{Is there any parallel-space efficient function, namely $\lambda(\omega)$, that can map threads only on the data elements of an embedded 2D fractal}? 
The second question relates to performance: \textit{Will the parallel-space improvement translate into a significant GPU performance improvement}?

The present work presents theoretical and experimental results that answer these two questions positively. A dedicated analysis is devoted to show that an alternating
unrolling strategy allows to define a parallel-space efficient $\lambda(\omega)$ that only requires $\mathcal{O}(n^\mathbb{H})$ threads, with $\mathbb{H}$ being the Hausdorff dimension of the fractal. In terms of performance, by taking advantage of intra-block parallelism, $\lambda(\omega)$ becomes computable in $\mathcal{O}(\log_2 \log_2 (n))$ time which is fast enough to produce a monotonically increasing speedup with respect to a bounding-box approach, once $n > n_0$ with $n_0$ being a threshold size different for each fractal. In addition to these results, the $\lambda(\omega)$ map is also adapted to GPU tensor core computation, further increasing its performance by up to $40\%$. This last result serves as an evidence that GPU tensor cores may be utilized in more ways than what they were initially thought for (deep learning).

This work is an extension and generalization of a previous conference work \cite{8291959} where preliminary results were presented for a specific fractal. This work generalizes the map for a family of embedded 2D fractals and includes an extensive presentation of experimental results with relevant tests, as well as a new section in which the proposed mapping is further accelerated by encoding the expression into tensor core operations. The rest of the manuscript is organized as follows: Section \ref{sec:related-work} presents related work on the field, Section \ref{sec:NBB-fractals} characterizes the NBB fractals family, Section \ref{sec:formulation-lambda} formulates the map with new theoretical results, Section \ref{sec:intra-block-mapping} describes several approaches for intra-block mapping, Section \ref{sec:case-study-sierpinski} offers the case study of the Sierpinski gasket with experimental performance measurements and Section \ref{sec:discussion-conclusions} concludes the work highlighting the main results, discussing them and describing possible further work within this line of research.

\section{Related Work}
\label{sec:related-work}
Jung \textit{et al.} \cite{Jung2008} explored the possibilities of improving the GPU mapping on triangular domains, by proposing packed data structures that represent triangular and symmetric matrices with applications to LU and Cholesky decomposition. Their strategy is based on building a \textit{rectangular box} for accessing and storing a triangular matrix (upper or
lower). Data structures become practically half the size with respect to classical methods based on the full matrix. The strategy was originally intended for saving memory (\textit{i.e.,} the matrix memory usage), however one can apply the concept analogously to save parallel space.

Ries \textit{et al.} contributed with a parallel GPU method for the triangular matrix inversion \cite{Ries:2009:TMI:1654059.1654069}.  The authors identified
that the parallel space indeed can be improved by using a \textit{recursive partition} of the grid\footnote{A grid is a collection of thread-blocks which are spatially organized and execute asynchronously one from another.}, based on a \textit{divide and conquer} strategy.  The
mapping approach takes $O(\log_2(n))$ with $n$ being the side of a square matrix.

Navarro, Hitschfeld and Bustos have proposed a block-space map function for $2$-simplices\footnote{A $k$-simplex is the generalization of the notion
of a triangle to $k$-dimensions. A $2$-simplex corresponds to the triangle while a $3$-simplex corresponds to a tetrahedron.}  and
$3$-simplices \cite{DBLP:conf/hpcc/NavarroH14, CLEI-2016-navarro, navarro2018competitiveness}, 
based on the solution of an $m$ order equation that is formulated from the linear enumeration of the discrete elements. The authors report performance
improvement for $2$-simplices, and for the $3$-simplex case, the mapping technique is extended to the discrete orthogonal tetrahedron, where the parallel space
usage can be $6\times$ more efficient. However the authors clarify that it is difficult to translate such space improvement into performance improvement, as the
map requires the computation of several square and cubic roots that introduce a significant amount of overhead to the process. From the point of view of
data-reorganization, a succinct blocked approach can be combined along with the block-space thread map, producing additional performance benefits with a
sacrifice of $o(n^3)$ extra memory. 

Exploring the benefits of efficient GPU mapping onto embedded 2D fractals is a relevant yet unexplored topic of research, as its geometry is no longer Euclidean as in the related works.  Finding a proper efficient $\lambda(\omega)$ would produce an asymptotic improvement in parallel space and a potential performance improvement that could eventually be exploited.

\section{Characterizing Discrete Embedded 2D Fractals}
\label{sec:NBB-fractals}
This Section characterizes embedded 2D fractals and defines two Lemmas regarding dimension and space packing, which provide useful insights to formulate an efficient GPU thread map for a family of embedded 2D fractals that share the same construction principle. Also, a block-space mapping strategy is proposed to further improve on the number of map computations performed.

\subsection{The Non-overlapping Bottom-up Boxes (NBB) Family}
Embedded 2D Fractals are discrete non-Euclidean structures that live in $\mathbb{Z}^2$ and are contained inside a tight Euclidean embedding space, \textit{i.e.}, a 2D bounding box. By being discrete structures, these fractals have a lower-bound when scaling down, \textit{i.e.}, a unit of space, but can scale-up infinitely. Because of this, discrete embedded 2D fractals are best described using a bottom-up approach rather than a top-down one. The bottom-up approach consists of defining the fractal as replications of itself from the previous level of scale, with different translations to each replica. An additional restriction is introduced to this type of fractals, which is that the bounding boxes\footnote{Not to be confused with the bounding-box approach used in GPU computing, which refers to a programming mode that uses the identity $f(x) = x$ to map from parallel space to data space.} of the replicas cannot overlap in space, that is, a location in the embedding space cannot be occupied by more than one replica. The rest of the manuscript will refer to this kind of fractals as \textit{Non-overlapping Bottom-up Boxes} fractals, or NBB fractals. 

Several fractals can be built following this scheme, including the Sierpinski Gasket, Cantor set, Vicsek fractal, H-Fractal, among others. Table \ref{tbl:discrete-fractals} presents a list with examples of NBB fractals with their bottom-up building step and their Hausdorff dimension as well.

\begin{table}[ht!]
    \begin{center}
    \begin{tabular}{ p{2.0cm} | c | l | p{2.2cm} }
    Fractal Name & Illustration & NBB Step & Hausdorff Dimension ($\mathcal{H} = \frac{\log(k)}{\log(s)}$)\\
    \hline
    Sierpinski Gasket &
    \raisebox{-\totalheight}{\includegraphics[scale=0.03]
    {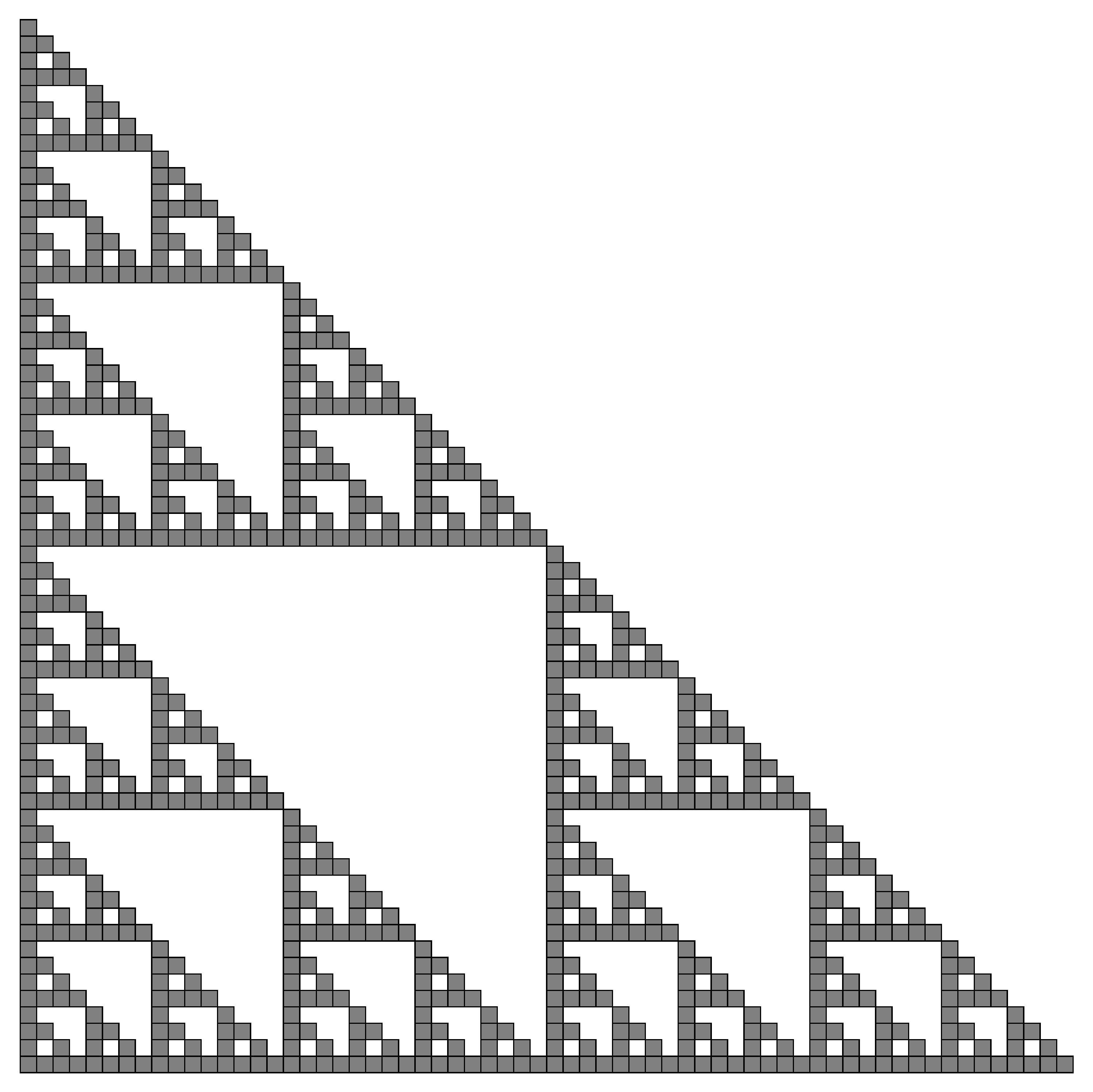}}  & 
    \raisebox{-\totalheight}{\includegraphics[scale=0.5]
    {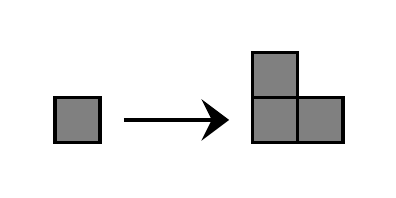}}  & 
    \ \newline $\frac{\log(3)}{\log(2)} \approx 1.58$\\
    \hline
    Chandelier (Custom) &
    \raisebox{-\totalheight}{\includegraphics[width=2.0cm, height=0.7cm]
    {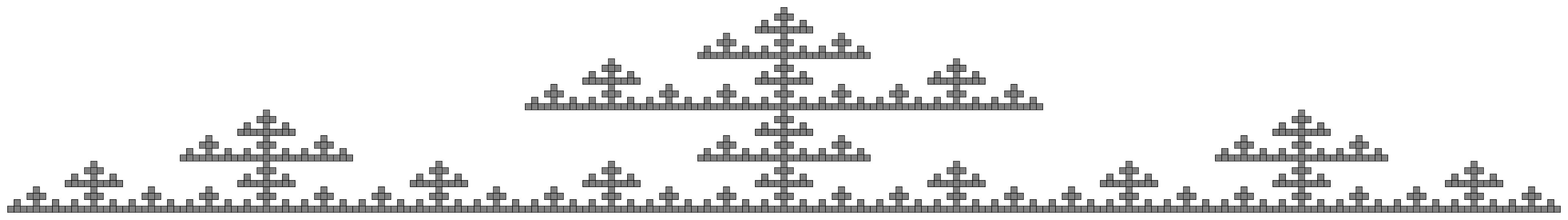}} & 
    \raisebox{-\totalheight}{\includegraphics[scale=0.5]
    {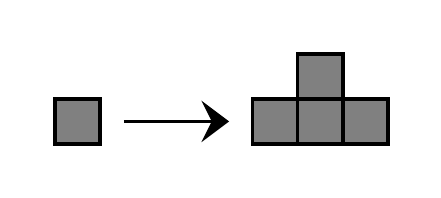}}  & 
    \ \newline $\frac{\log(4)}{\log(3)} \approx 1.26$\\
    \hline
    H-Fractal &
    \raisebox{-\totalheight}{\includegraphics[scale=0.03]
    {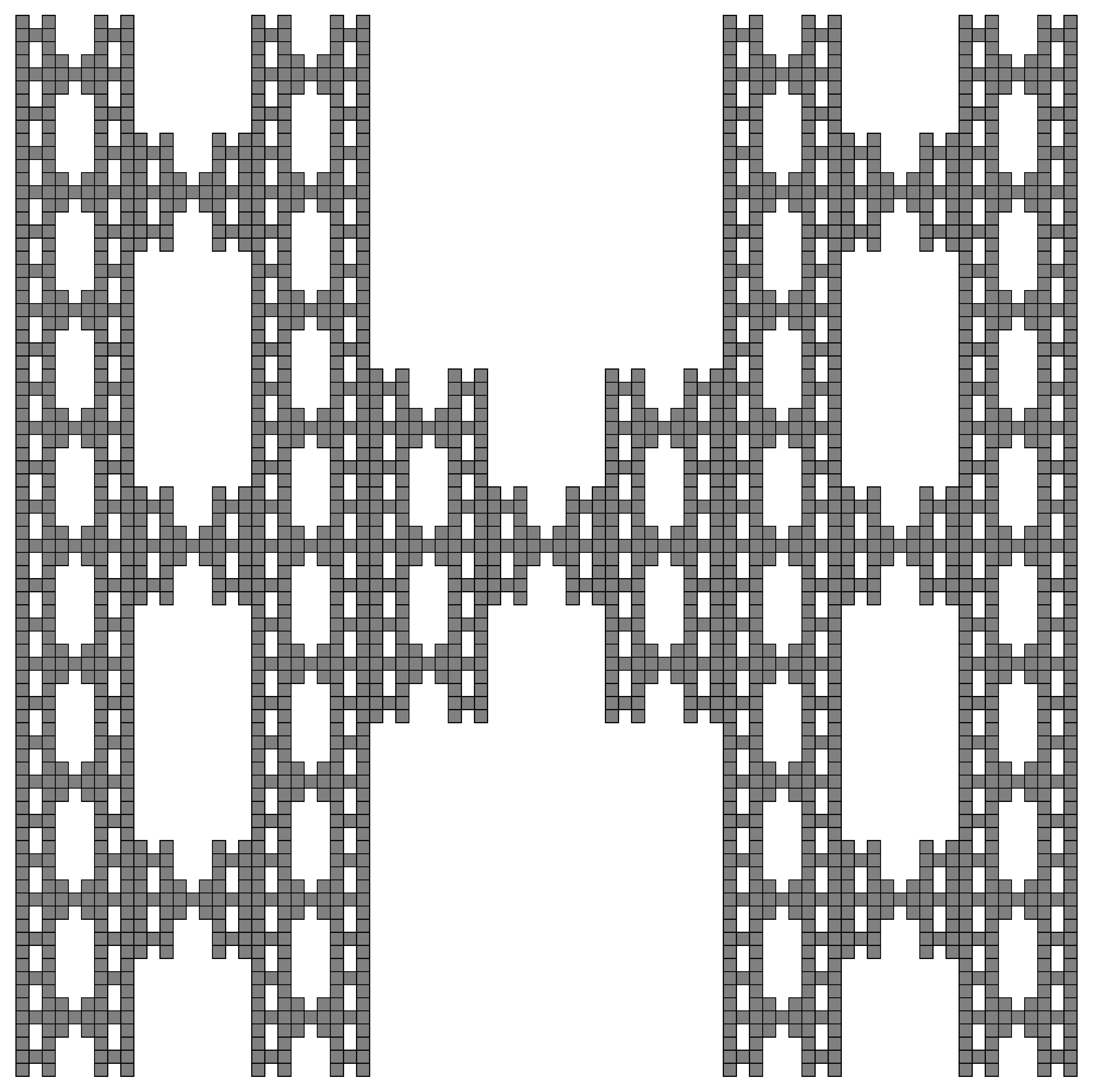}} & 
    \raisebox{-\totalheight}{\includegraphics[scale=0.5]
    {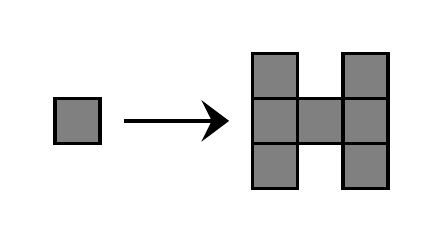}}  & 
    \ \newline $\frac{\log(7)}{\log(3)} \approx 1.77$\\
    \hline
    Candy (Custom) &
    \raisebox{-\totalheight}{\includegraphics[scale=0.012]
    {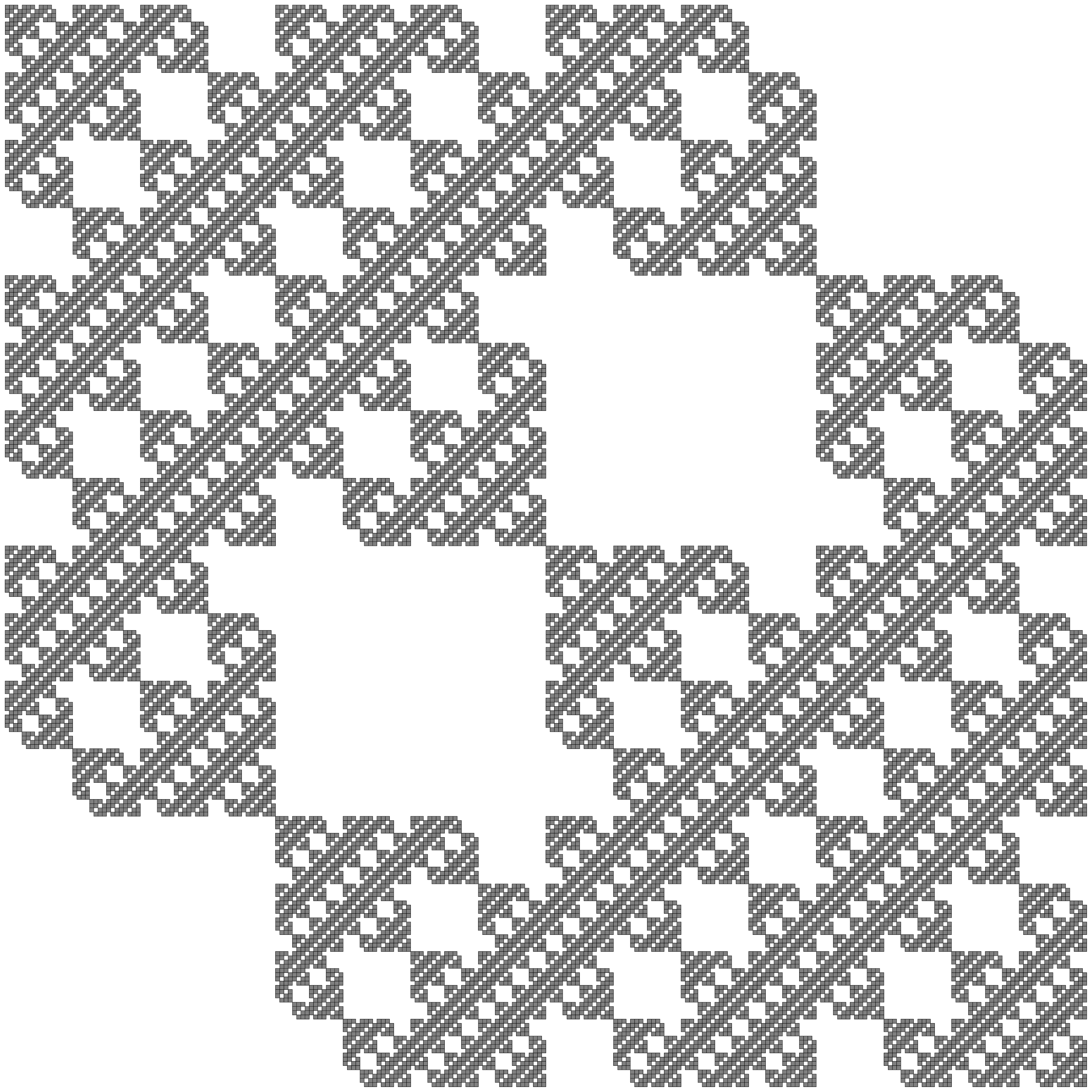}} & 
    \raisebox{-\totalheight}{\includegraphics[scale=0.5]
    {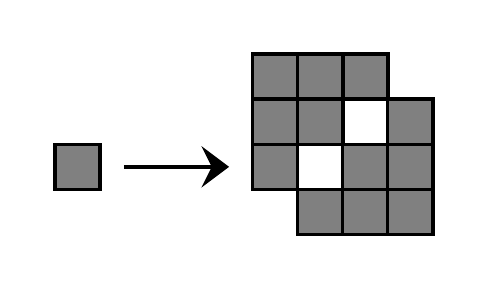}}  & 
    \ \newline $\frac{\log(12)}{\log(4)} \approx 1.79$\\
    \hline
    Sierpinski Carpet &
    \raisebox{-\totalheight}{\includegraphics[scale=0.035]
    {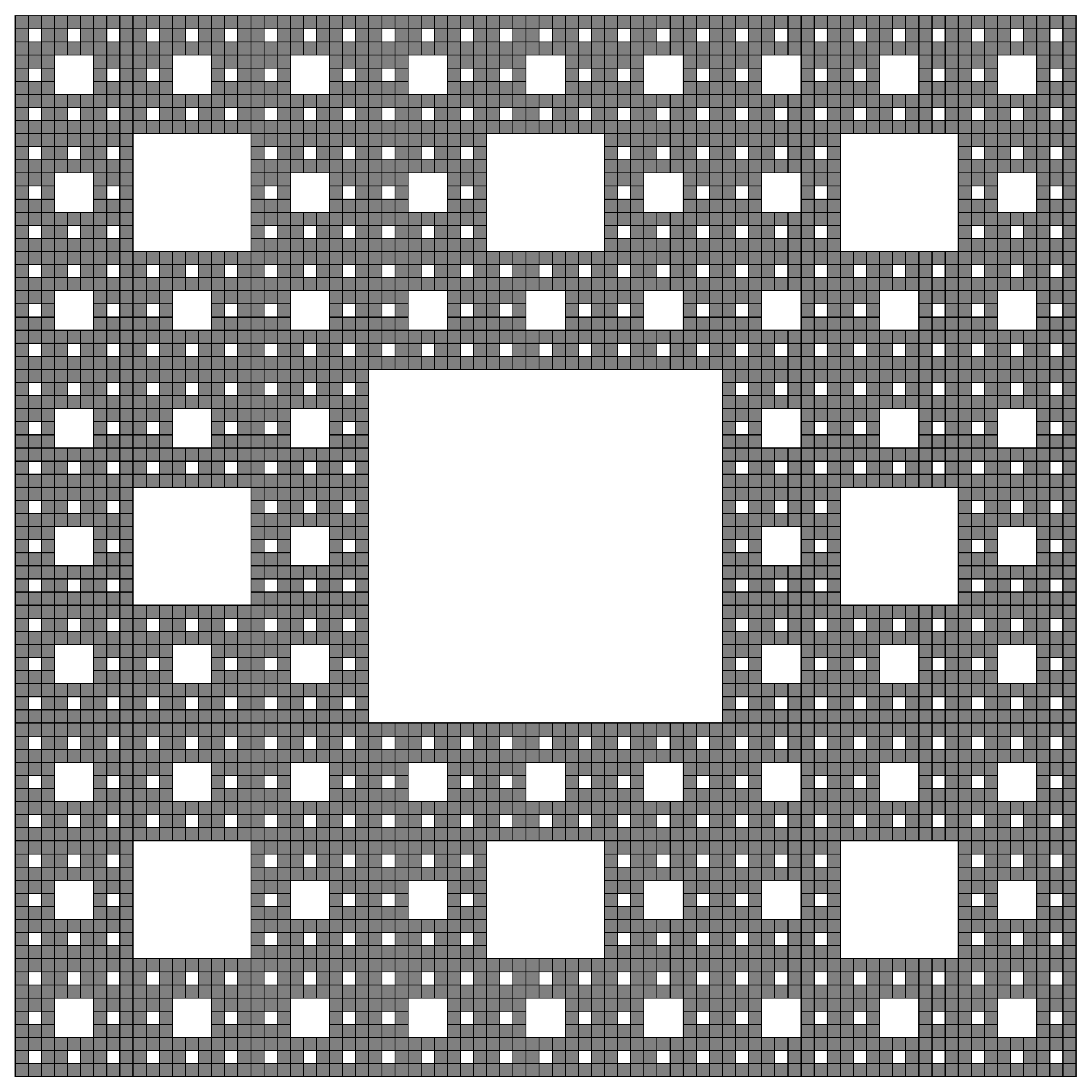}} & 
    \raisebox{-\totalheight}{\includegraphics[scale=0.5]
    {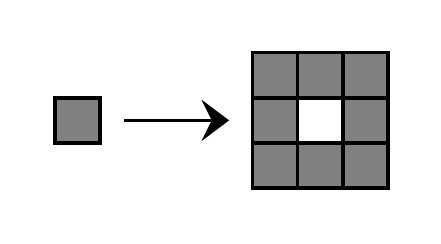}}  & 
    \ \newline $\frac{\log(8)}{\log(3)} \approx 1.89$\\
    \hline
    X-Fractal (Custom) &
    \raisebox{-\totalheight}{\includegraphics[scale=0.010]
    {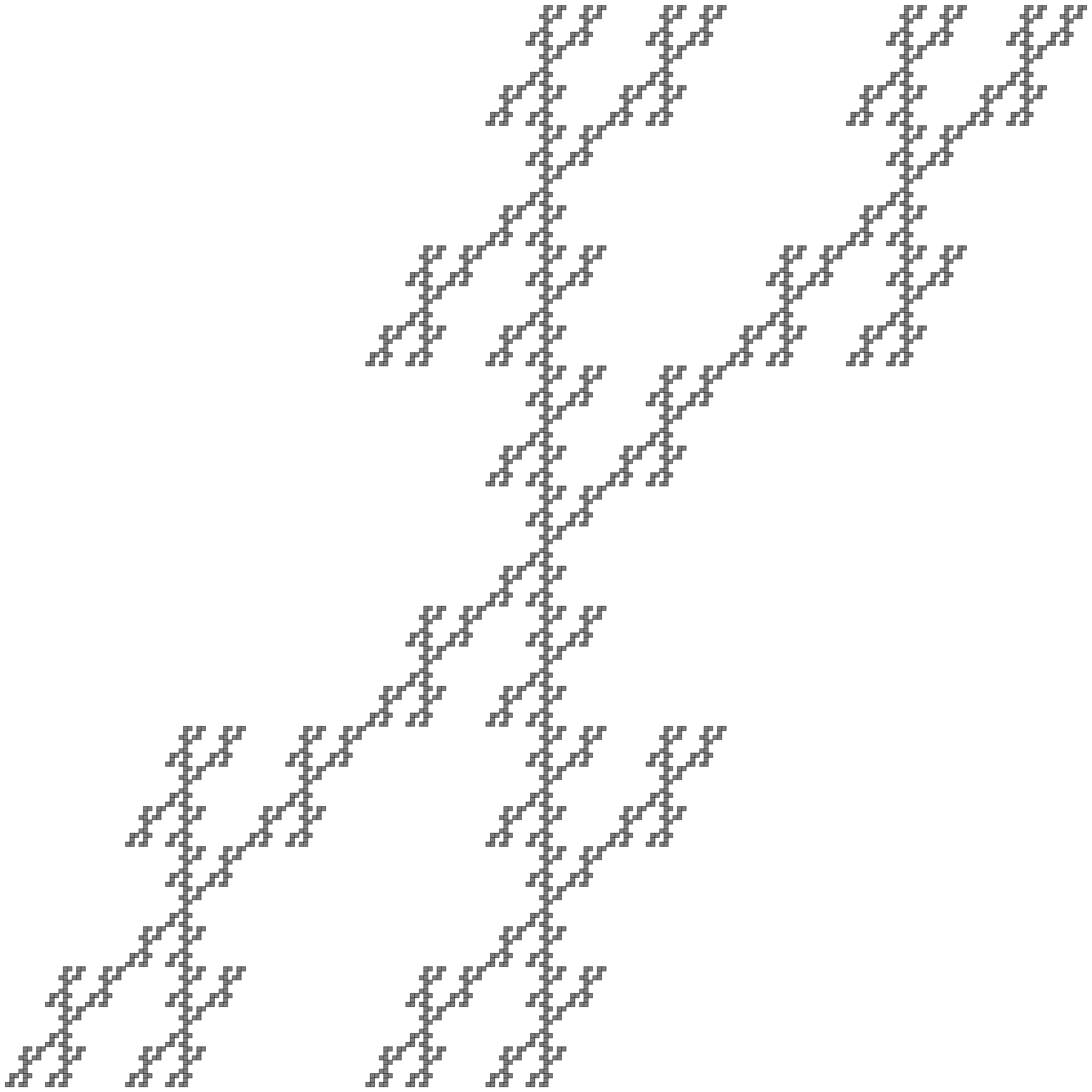}} & 
    \raisebox{-\totalheight}{\includegraphics[scale=0.5]
    {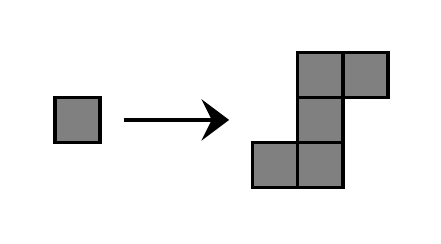}}  & 
    \ \newline $\frac{\log(5)}{\log(3)} \approx 1.46$\\
    \hline
    Vicsek Fractal &
    \raisebox{-\totalheight}{\includegraphics[scale=0.010]
    {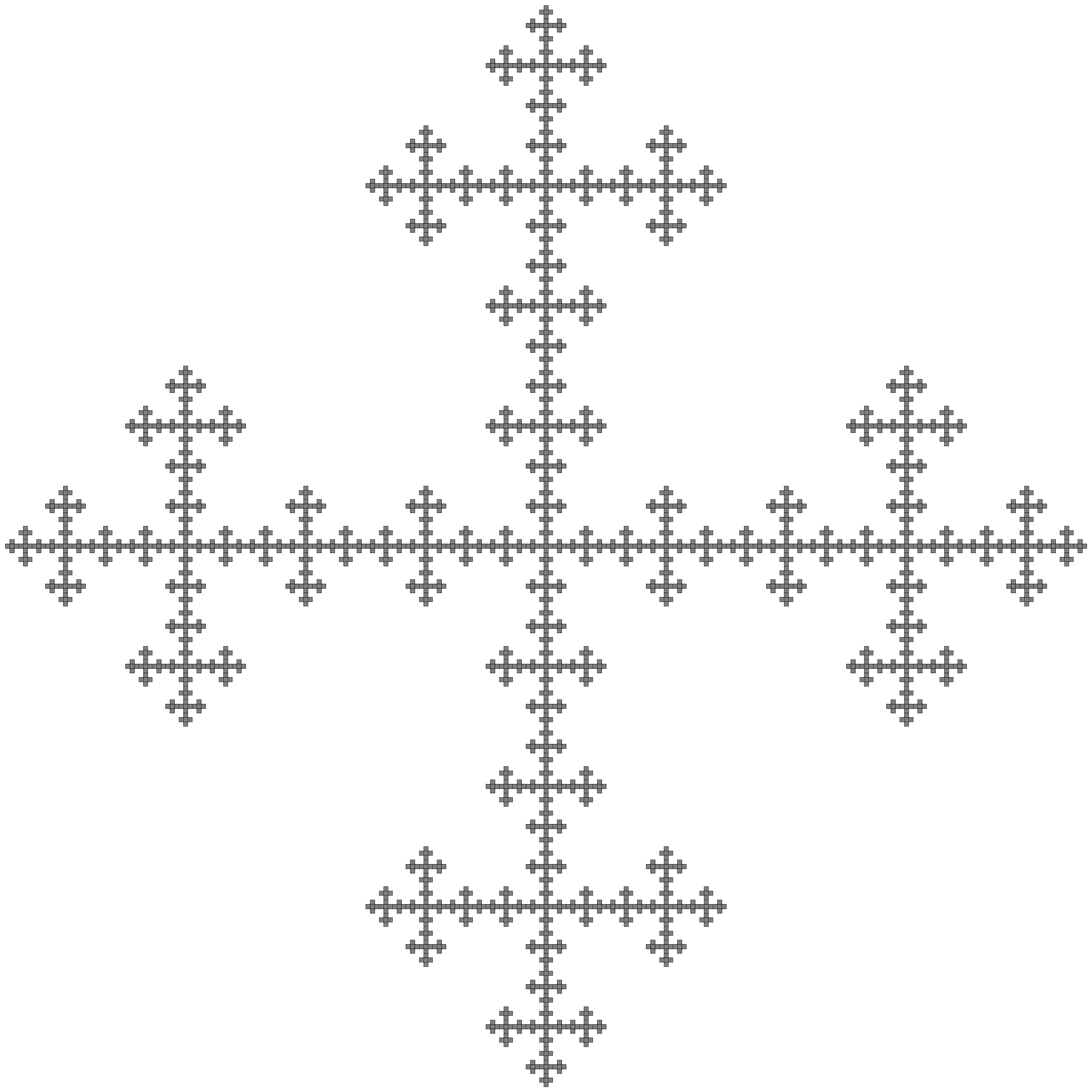}} & 
    \raisebox{-\totalheight}{\includegraphics[scale=0.5]
    {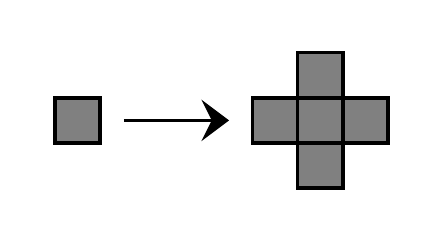}} & 
    \ \newline $\frac{\log(5)}{\log(3)} \approx 1.46$\\
    \hline
    Empty-Bottles (Custom) &
    \raisebox{-\totalheight}{\includegraphics[scale=0.03]
    {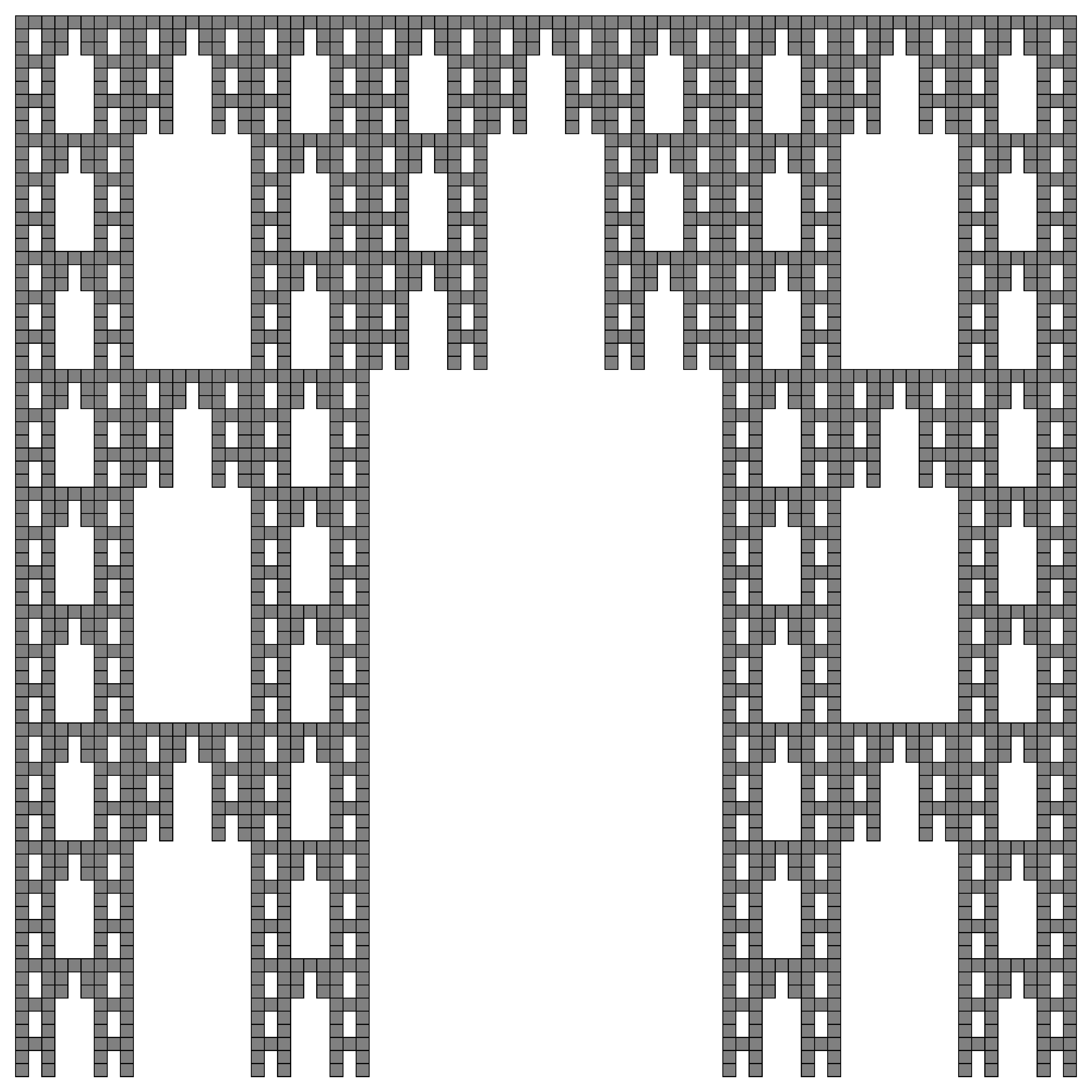}} & 
    \raisebox{-\totalheight}{\includegraphics[scale=0.5]
    {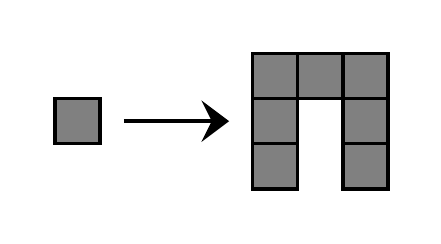}} & 
    \ \newline $\frac{\log(7)}{\log(3)} \approx 1.77$\\
    \hline
    Cantor set &
    \raisebox{-\totalheight}{\includegraphics[scale=0.010]
    {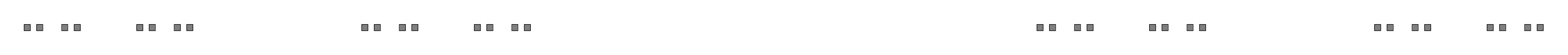}} & 
    \raisebox{-\totalheight}{\includegraphics[scale=0.5]
    {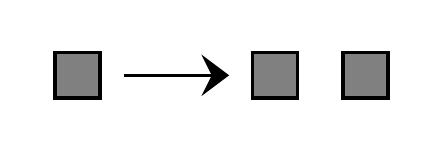}} & 
    \ \newline $\frac{\log(2)}{\log(3)} \approx 0.63$\\
    \hline
    \end{tabular}
    \caption{Example fractals of the NBB family.}
    \label{tbl:discrete-fractals}
    \end{center}
\end{table}

The notation $\fracnot{n}{k}{s}$ is introduced to denote an embedded 2D fractal of the NBB family, where $n \in \mathbb{N}$ is its linear size in one axis, $k \in \mathbb{N}$ the number of self-similar replicas for the next recursive step and $s \in \mathbb{N}$ the scale-up factor between a given scale level and the upcoming one, along each dimension\footnote{It is important to mention that these definitions work for irregular fractals as well, as in the case of the Chandelier fractal, where it can have two definitions; $\fracnot{n_x}{4}{3}$ or $\fracnot{n_y}{4}{2}$ along the $x$ and $y$ axis, respectively.}(for example, for the H-fractal $k=7, s=3$ and for the Candy fractal $k=12, s=4$). The space used by a fractal, denoted as $\mathcal{V}(\fracnot{n}{k}{s})$ may be expressed recursively as
\begin{equation}
    \mathcal{V}(\fracnot{n}{k}{s}) = \sum_{i=1}^{k} {\mathcal{V}_i(\fracnot{sn}{k}{s})} 
    \label{eq_fractal_description}
\end{equation}
with $\mathcal{V}(\fracnot{1}{k}{s}) = 1$ being the limit condition of the recursion. Since $k$ is fixed, and $n$ scales up by factors of $s$, the volume may expressed as
\begin{equation}
    \mathcal{V}(\fracnot{n}{k}{s}) = k^r
\end{equation}
where $r = \log_{s}(n)$ is defined as the scale level.

\subsection{Dimension and Packing of NBB Fractals}
The following Lemma guarantees that the dimensionality of NBB fractals, in their $\mathcal{V}(\fracnot{n}{k}{s}) = k^r$ form, is actually fractal.
\begin{lemm}
\label{lemma_hausdorff}
The space occupied by an NBB fractal is in correspondence with its Hausdorff dimension in the scale-up limit.
\end{lemm}
\begin{proof}
    The space occupied by an NBB fractal is 
    $\mathcal{V}(\fracnot{n}{k}{s}) = k^r$.  Given that $r =
    \log_{s}(n)$ and $k^{log_{s}(n)} = (s)^{\log_{s}(k)\log_{s}(n)}$, the space expression
    can be rearranged into
    \begin{equation} 
    \mathcal{V}(\fracnot{n}{k}{s}) = n^{\log_{s}(k) = \frac{\log(k)}{\log(s)}} = n^{\mathcal{H}}
    \end{equation}
    where the exponent $\mathcal{H}$ is the Hausdorff dimension, \textit{i.e.}, the quotient of the logarithm of the number of replicas and the logarithm of the scaling factor.
\end{proof}
Lemma \ref{lemma_hausdorff} guarantees that discrete embedded 2D fractals, which have a lower bound in scale and can only grow by scaling up, still exhibit their Hausdorff dimension when $n \mapsto \infty$. Another useful fact is that by having one GPU thread per unitary element of the fractal is already resource-efficient as it would yield a fractal space occupancy in the parallel space as well. The next Lemma relates the geometries of parallel-spaces with fractal domains.
\begin{lemm}
\label{lemma_regular}
    A NBB fractal $\fracnot{n}{k}{s}$ can pack into a $2$-orthotope $\Pi^2$ of dimensions $k^{\lceil \frac{r}{2} \rceil} \times k^{\lfloor \frac{r}{2} \rfloor}$ at any scale level $r$.
\end{lemm}
\begin{proof}
    By induction on $r$:
    \begin{itemize}
        \item{Base case:} At scale $r=0$ the fractal has a space of $\mathcal{V}(\fracnot{1}{k}{s}) = 1$ element that packs into a regular
            $2$-orthotope of $1 \times 1 = k^{\lceil \frac{0}{2} \rceil} \times k^{\lfloor \frac{0}{2} \rfloor }$ satisfying $k^{\lceil \frac{r}{2} \rceil} \times k^{\lfloor \frac{r}{2} \rfloor}$.
        \item{Induction step:} It is assumed that the orthotope at scale level $r$ is quasi-regular or regular. If $r$ is even, the packing for
            $r+1$ will scale by $k$ in the horizontal dimension of the $2$-orthotope. If $r$ is odd, the packing for $r+1$ will scale by $k$ in the vertical dimension of $\Pi^2$. Since even and odd must alternate, the dimensions of the packed $2$-orthotope $\Pi^2$ for $r+1$ can only be $k \cdot k^{\lceil{\frac{r}{2}}\rceil}  \times k^{\lfloor{\frac{r}{2}}\rfloor}$ for even $r$, or $k^{\lceil{\frac{r}{2}}\rceil} \times k\cdot k^{\lfloor{\frac{r}{2}}\rfloor}$ for odd $r$, 
            which is quasi-regular or regular, respectively.
    \end{itemize}
\end{proof}
In GPU Computing the parallel space (grid, blocks, threads) can only be defined as Euclidean boxes in 1D, 2D or 3D, which becomes a constraint when processing NBB fractals. In this context, Lemma \ref{lemma_regular} gives useful insights on how one could map the Euclidean space onto the Fractal one.

\subsection{Changing from Thread-space to Block-space}
An important aspect to consider is at which level the parallel space will be mapped. Two approaches are possible; (1) thread-space mapping and (2) block-space mapping. Let $\lambda(\omega)$ be the map from GPU parallel-space (Euclidean) to embedded 2D fractal space. For the first approach, $\lambda(\omega)$ defines $\omega$ as a unique thread location in parallel space.  For the second approach, $\lambda(\omega)$ defines
$\omega$ as a block coordinate in which several threads are contained. The block-space approach has three important advantages over thread-space mapping. First, in block-space, the fractal becomes a coarsened version of the original, requiring fewer elements to be mapped. Second, since the fractal is a simplified
version of itself, it is possible to work on higher sizes of $n$ before the CUDA grid maximum dimensions or numerical limits are reached. Third, the block-space approach allows the
possibility for threads inside a block to preserve locality, which is essential for doing efficient coalesced memory accesses on GPU global memory.

The next Section formulates $\lambda(\omega)$ with $\omega = (\omega_x,\omega_y)$ being the two-dimensional block coordinate of constant size $|B| = \rho \times \rho$ threads.  The change from thread-space to block-space means that blocks are mapped to a simplified version of the fractal of linear size $n_b = n/b$ with $b=\rho$.

\section{Formulation of GPU map $\lambda(\omega)$}
\label{sec:formulation-lambda}
The function $\lambda: \mathbb{Z}_{\mathbb{E}}^{2} \mapsto
\mathbb{Z}_{\mathbb{F}}^{2}$ is introduced as a mapping of block coordinates $\omega$ from GPU parallel-space $\Pi^2$, which lies in Euclidean space $\mathbb{Z}_\mathbb{E}$, onto block coordinates in the embedded fractal space $\mathbb{Z}_{\mathbb{F}}$. The intuition behind the formulation of $\lambda(\omega)$ is an 
unrolling process applied in parallel to each $\omega \in \Pi^2$ through all the scale levels of the fractal in the scale-down direction until the unit scale limit is reached. At each scale level, different $\Delta_x,\Delta_y$ offsets are accumulated to form the final $(\lambda_x(\omega), \lambda_y(\omega))$ coordinate in the embedded domain of the fractal.
\begin{thma}
\label{theorem_lambda}
    There exists $\lambda(\omega)$ that maps a GPU parallel-space of size $|\Pi^2| = \mathcal{O}(n^\mathcal{H})$ to any NBB fractal in $\mathcal{O}(\log_2 \log_2 (n_b))$ time using $|B| = \mathcal{\theta}(\frac{\log_2(n_b)}{\log_2 \log_2(n_b)})$ threads per block. 
\end{thma}
\begin{proof}
    \textit{By construction}: let $r_b = \log_{s}({n_b})$ be the block-space scale level of the fractal, $\Pi^2$ the $2$-orthotope of $k^{\lceil \frac{r_b}{2} \rceil} \times k^{\lfloor \frac{r_b}{2} \rfloor}$ blocks that maps onto the discrete embedded 2D fractal $\fracnot{n_b}{k}{s}$, with each block having
    $b \times b$ threads. By Lemma (\ref{lemma_hausdorff}), $\Pi^2$ is parallel-space efficient in block-space, \textit{i.e.}, $|\Pi^2| = \mathcal{O}(n^\mathcal{H})$. A helper index function $\beta_\mu(\omega)$ is defined as
    \begin{equation}
        \beta_\mu(\omega) = \Big( \frac{\omega_x(\mu \mod 2) + \omega_y((\mu+1 )\mod 2)}{k^{\lceil \frac{\mu}{2} \rceil-1}}\Big) \mod k
    \end{equation}
    to generate indices in the range $\beta_\mu(\omega) \in [0, k-1]$ that identifies, within scale level $\mu \in [0..r_b]$, which of the $k$ regions of the fractal does block $\omega$ belongs to. For even $\mu$, $\beta_\mu(\omega)$ acts on $\omega_x$. For odd $\mu$, it acts on $\omega_y$. 
    
    An arbitrary numbering is chosen for associating the $k$ blocks of the fractal's NBB step with the $k$ different $\beta_{\mu}(\omega)$ values\footnote{For example, for the Vicsek fractal, where $k=4$, the regions can follow a numbering of the form top (0), bottom (1), left (2) and right (3). For the Sierpinski gasket, where $k=3$, a valid numbering could be top (0), bottom (1) and right (2). In the same way, the Candy fractal would require a numbering for its $k=12$ replicas.}. A perfect hash table\footnote{The hash table may be replaced by an arithmetic expression, yielding the same values.} $H[\ ]$ of size $k$ can be used to map the $k$ values of $\beta_\mu(\omega)$ to $(\tau_x^\mu,\tau_y^\mu)$ replica offsets of the form
    \begin{align}
        \tau^{\mu} = H[\beta_\mu(\omega)] = (\tau_x^\mu, \tau_y^\mu),\ \  \tau_x^u, \tau_y^u \in [0..s-1].
    \end{align}
    The replica offsets combined with the corresponding fractal replica side length $k^{\mu-1}$, gives the corresponding offset in embedded space
    \begin{align} 
        \Delta^\mu = (\tau_x^{\mu}(s)^{\mu-1}, \tau_y^{\mu}(s)^{\mu-1}) = (\Delta_x^\mu, \Delta_y^\mu)
        \label{eq:deltas}
    \end{align}
    that contributes to the final mapped coordinate. The summation of all partial coordinates produces the map 
    \begin{align}
        \lambda(\omega) &= (\lambda_x(\omega), \lambda_y(\omega)), 
        \label{eq:lambda}\\
        \lambda_x(\omega) &= \sum_{\mu=1}^{\log_{s}(n_b)} \Delta_x^\mu
        \label{eq:lambda-x}\\
        \lambda_y(\omega) &= \sum_{\mu=1}^{\log_{s}(n_b)} \Delta_y^\mu
        \label{eq:lambda-y}
    \end{align}
    which can be computed in $\mathcal{O}(\log_2\log_2(n))$ time (\textit{i.e.,} $n_b \in \theta(n)$) using a parallel reduction with the threads contained
    in the $\omega$ block. Finally, by Brent's Theorem \cite{Brent:1974:PEG:321812.321815}, $|B| = \mathcal{\theta}\Big(\frac{\log_2(n)}{\log_2 \log_2(n)}\Big)$ threads are sufficient for a block of threads to reduce efficiently in parallel. 
\end{proof} 
Theorem \ref{theorem_lambda} guarantees the existence of an efficient $\lambda(\omega)$ map for any NBB fractal. It is important to mention that the hash table is of fixed size $k$ and the same table is reused at every scale level. In practice this hash table may be defined at compile time as a static resource, or as a GPU shared memory constant array for a whole block of threads. For some fractals it is possible to replace the hash table for an arithmetic hash function that returns the replica offsets directly.

\begin{thma}
\label{theorem_speedup}
    Processing a NBB fractal with $\lambda(\omega)$ requires asymptotically less work than using a bounding-box approach. 
\end{thma}
\begin{proof}
   %
    The asymptotic work improvement factor of $\lambda(\omega)$ with respect to the bounding-box approach is the quotient of the costs of mapping all
    blocks using their corresponding $\Pi^2$ structures with the consideration $n_b \in \theta(n)$
    \begin{align}
        S_{\lambda(\omega)} &=      \frac{\mathcal{O}(1) \mathcal{V}(\Pi_{BB}^2)}{\mathcal{O}(\log_2\log_2(n))\mathcal{V}(\Pi^2_{\lambda(\omega)})}\\
    \end{align}
    where $\Pi^2_{BB}$ and $\Pi_{\lambda(\omega)}^2$ are the parallel-spaces for the bounding-box and $\lambda(\omega)$ approaches, respectively. The
    parallel-space of $\Pi_{BB}^2$ corresponds to the Euclidean box of $n_b \times n_b$ blocks, and the parallel-space of $\Pi_{\lambda(\omega)}^2$ is $\mathcal{O}(n^{\mathcal{H}})$ by Lemma (\ref{lemma_hausdorff}). Applying the limit $n\to\infty$ gives
    \begin{align}
        \lim_{n\to\infty}{S_{\lambda(\omega)}} &= \lim_{n\to\infty} {\frac{\frac{\partial}{\partial n}(n^{2-\mathcal{H}})}{\frac{\partial }{\partial n}(\log_2\log_2(n))}}\\
                                               &= \lim_{n\to\infty} {\frac{(2-\mathcal{H})n^{1-\mathcal{H}}}{\frac{1}{n\log_2(n)}}} = \infty
    \end{align}
\end{proof}
    The importance of Theorem (\ref{theorem_speedup}) is that it guarantees the existence of a fractal size $n > n_0$ where the speedup provided by $\lambda(\omega)$ behaves as a monotonically increasing function. The smaller the Hausdorff dimension of the fractal, the stronger the behavior. The next Section covers the possible approaches to handle intra-block mapping, \textit{i.e.}, how threads inside a block can access individual fractal locations reached by the block-space mapping.

\section{Intra-Block Mapping}
\label{sec:intra-block-mapping}
Once $\lambda(\omega)$ maps a block $\omega$, all of its threads contained share 
the same block-space mapped coordinate in embedded space which serves as a 
reference location for each thread to compute their individual location in 
the fractal. This phase of organizing the threads within a block is defined 
here as \textit{Intra-Block Mapping}, and this Section describes three 
possible approaches to accomplish it.

\subsection{Further Unrolling}
In this approach threads inside their mapped block may use the same
$\lambda(\omega)$, with the same hash table or arithmetic hash function, 
but this time applied to each thread in local space. By Theorem (\ref{theorem_lambda}), the \textit{Intra-block map} is still parallel-space efficient and the mapping time 
becomes $\mathcal{O}(\log_2 \log_2(|B|)) \in \mathcal{O}(1)$ as the size $\rho \times \rho$ of a block is constant.

\subsection{Shared Lookup Table}
This second approach is to use a shared lookup table of size $\rho \times \rho =
\mathcal{O}(1)$ holding the final offset coordinates for each thread within the same block. Mapping each thread would cost $\mathcal{O}(1)$ memory accesses and the 
extra memory introduced by the shared table is $\mathcal{O}(\rho \times \rho) 
\in \mathcal{O}(1)$.

\subsection{Bounding Sub-boxes}
The third approach consists of using the mapped blocks as bounding sub-boxes. 
This approach introduces a constant number of extra threads in each block, 
but allows each thread to be mapped just with $f(x) = x$ which costs
$\mathcal{O}(1)$. If this method is chosen, then threads require a fast method 
to know if they belong to the fractal or not. 

Regardless of which Intra-block mapping approach is chosen, the final mapping time will not surpass the $\mathcal{O}(\log_2\log_2(n))$ time, as the blocks have
a constant size of threads, regardless of the value of $n$. Still, it is worth considering the differences in the approaches; \textit{Further Unrolling} introduces a constant cost in mapping time, the \textit{Shared Lookup Table} approach introduces a constant cost in memory and the \textit{Bounding Sub-boxes} introduce a constant in the number of extra threads. Choosing one or another can depend on the specific application, \textit{i.e.}, to avoid competing with the application in the use of memory bandwidth or arithmetic operations.

\section{Case Study: The Sierpinski Gasket}
\label{sec:case-study-sierpinski}
This Section applies the formulations and approaches from Sections \ref{sec:NBB-fractals} and \ref{sec:formulation-lambda}, which were generic to all NBB fractals, now for the specific case of the Sierpinski gasket. Experimental performance results are presented for different test cases (involving different compute patterns that are frequently found in discrete simulations), using different fractal sizes of the Sierpinski gasket.

The \textit{Sierpinski Gasket}, illustrated in Figure \ref{fig_sierpinski_discrete_construction_steps}. was described by Waclaw Sierpinski in 1915. 
\begin{figure}[ht!]
\centering
\includegraphics[scale=0.40]{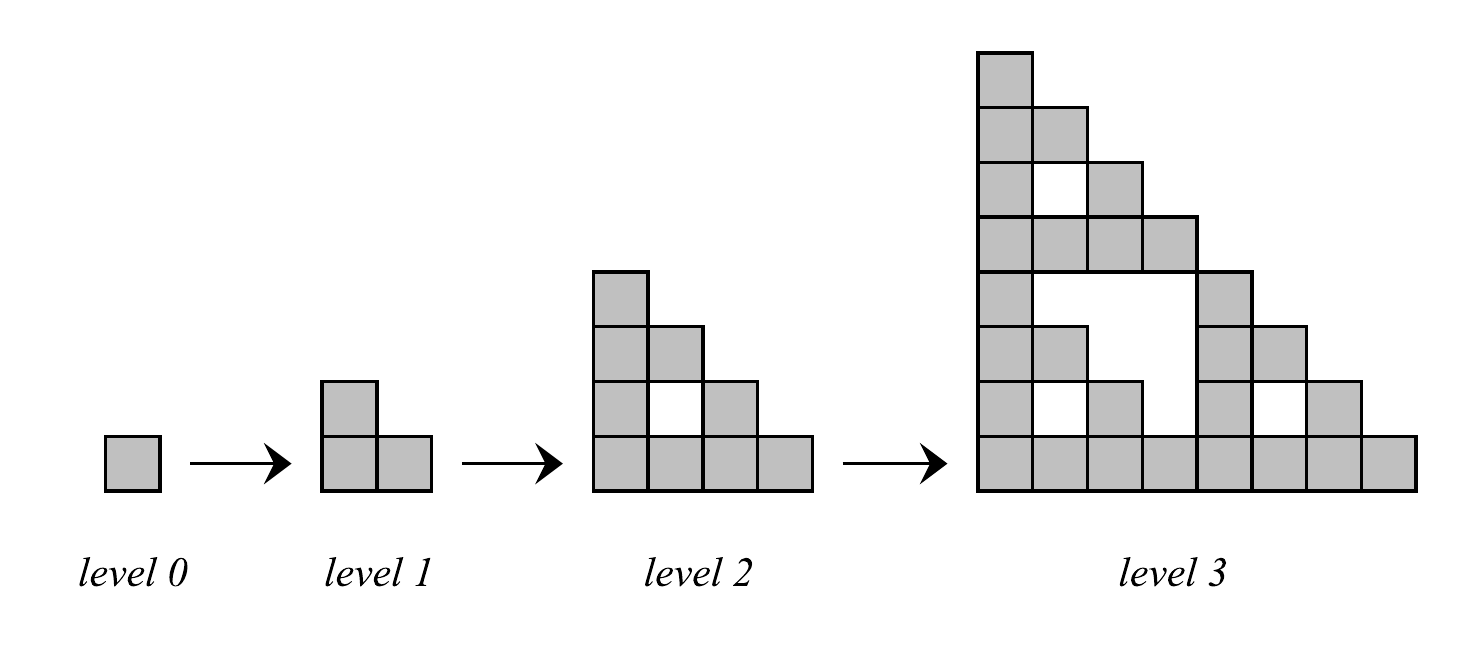}
\caption{Bottom-up construction of the discrete Sierpinski gasket.}
\label{fig_sierpinski_discrete_construction_steps}
\end{figure}

Being over a century old, this NBB fractal is still relevant as it is object of study in different fields such as the construction of antennas \cite{855489, 664115}, cellular automata \cite{Ohi2001, RevModPhys.55.601}, fractal molecular assembly \cite{shang2015}, DNA self-organization \cite{rothemund2004}, self-assembly theory \cite{Doty:2012:TAS:2380656.2380675, LATHROP2009384} and phase transitions on fractal spin lattices \cite{0305-4470-17-2-028, PhysRevLett.45.855, Mota20086095}, among others.  The Sierpinski gasket is denoted $\fracnot{n}{3}{2}$, where $k=3$ and $s=2$.   

\subsection{Defining $\lambda(\omega)$ for the Sierpinski Gasket}
The packing process of Lemma (\ref{lemma_regular}) describes an unrolling process, in which each block of threads, with coordinate $\omega$ in parallel space, accumulates a series of offsets to return a final mapped block coordinate in embedded fractal space. The specific case of the Sierpinski Gasket is illustrated in Figure \ref{fig_sierpinski_packing} where the unrolling principle is visible by the different shades that are in correspondence with the shaded replicas of the fractal. 
\begin{figure}[ht!]
    \centering
    \includegraphics[scale=0.15]{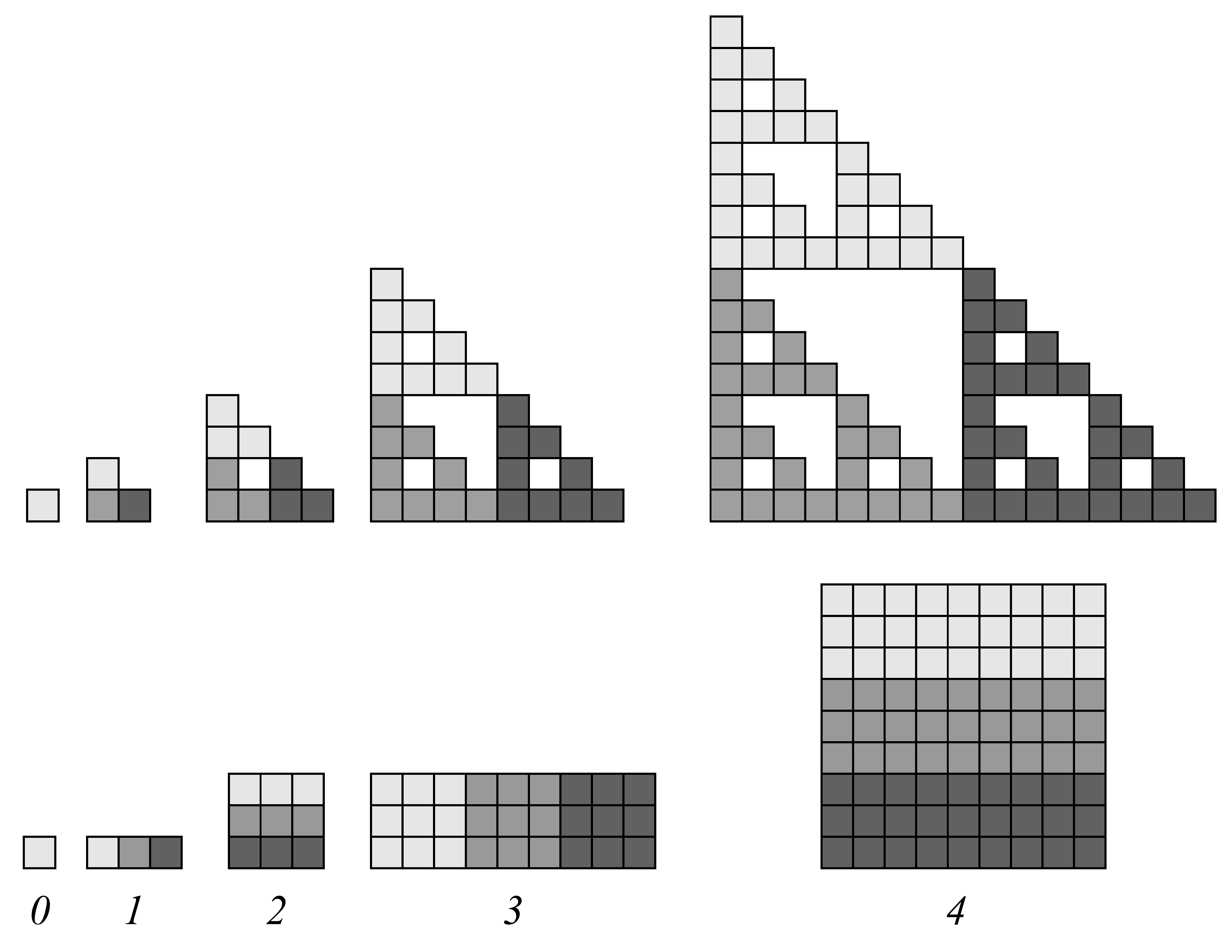}
    \caption{Each scale of the Sierpinski fractal packs into a $2$-orthotope
    $\Pi^2$ of dimensions $3^{\lceil \frac{r}{2} \rceil} \times 3^{\lfloor \frac{r}{2} \rfloor}$.}
    \label{fig_sierpinski_packing}
\end{figure}

The helper parameter $\beta_u$, for the case of the Sierpinski gasket, is
\begin{equation}
    \beta_\mu(\omega) = \Big( \frac{\omega_x(\mu \mod 2) + \omega_y((\mu+1 )\mod 2)}{3^{\lceil \frac{\mu}{2} \rceil-1}}\Big) \mod 3.
\end{equation}
Replica regions 
are numbered as $0$ (top), $1$ (middle) and $2$ (right) (see Figure \ref{fig_sierpinski_packing}, top, for visual reference). 
The hash table for this fractal is $H[0] = (0,0), H[1] = [0,1], H[2] = [1,1]$ where each pair is the corresponding replica offset. In the case of the Sierpinski gasket, it is also possible to use the following arithmetic hash function
\begin{align}
    h(\beta_\mu) = (\tau_x^\mu, \tau_y^\mu) = (\Big\lfloor \frac{\beta_\mu}{2} \Big\rfloor, \beta_\mu - \Big\lfloor \frac{\beta_\mu}{2} \Big\rfloor)
\end{align}
as an alternative to the hash table, giving the same replica offsets for each of the $x$ and $y$ directions at scale level $\mu$. The replica offsets are combined with the replica linear sizes to form the offsets in embedded space
\begin{align} 
    \Delta^\mu = (\Delta_x^\mu, \Delta_y^\mu) = (\tau_x^{\mu}2^{\mu-1}, \tau_y^{\mu}2^{\mu-1})
\end{align}
Having defined the required functions and parameters, the $\lambda(\omega)$ map for the Sierpinski gasket becomes
\begin{align}
\lambda(\omega) &= 
\Bigg(\sum_{\mu=1}^{r_b} \Delta_x^\mu, \sum_{\mu=1}^{r_b} \Delta_y^\mu \Bigg)
\label{eq:lambda-sierpinski}
\end{align}
with $r_b = \log_2(n_b)$.
By Theorem \ref{theorem_speedup}, $\lambda(\omega)$ is asymptotically faster than a bounding box approach. The theoretical parallel space improvement as well as the speedup are presented in Figure \ref{fig_theoretical_improvements}.

\begin{figure}[ht!]
\centering
\includegraphics[scale=0.72]{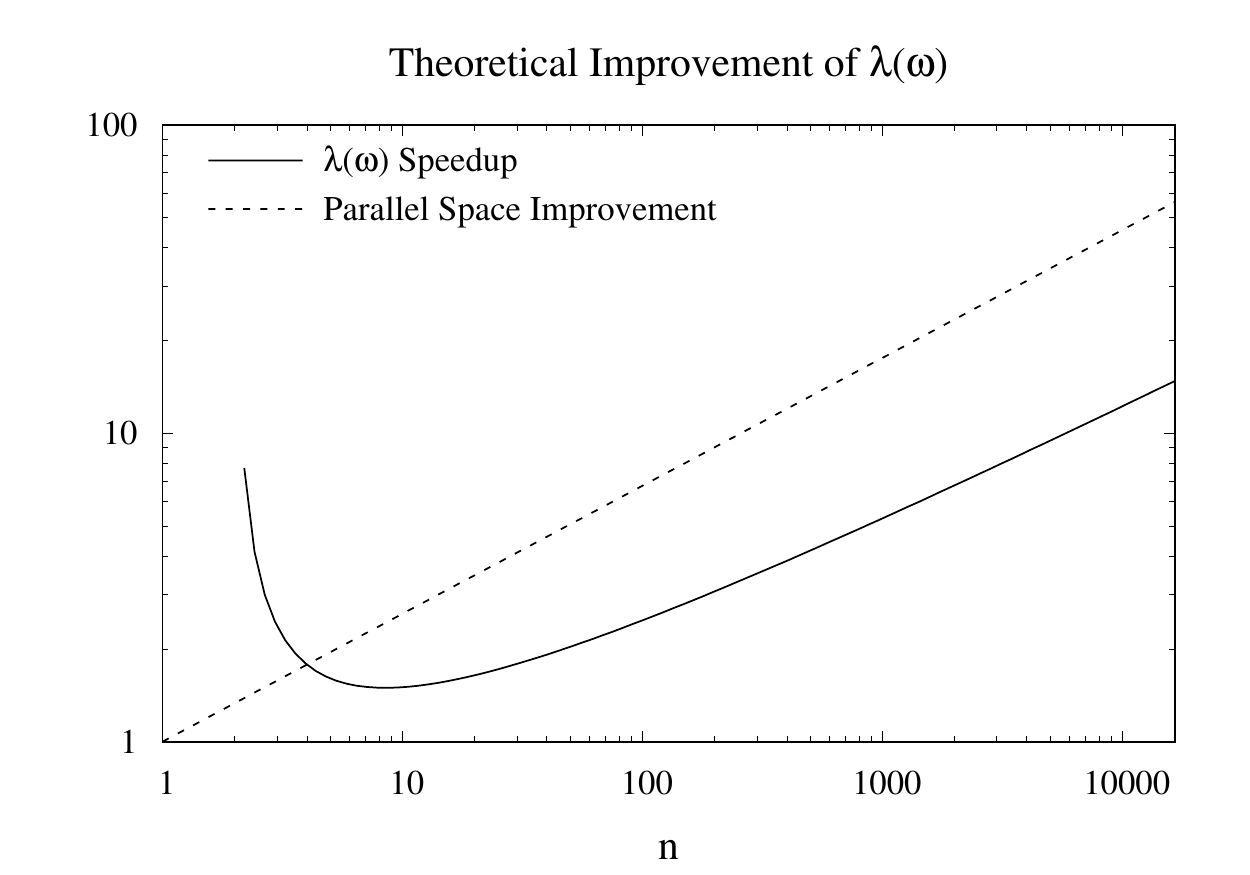}
    \caption{Theoretical improvement for parallel-space and mapping speedup for the Sierpinski gasket.}
\label{fig_theoretical_improvements}
\end{figure}
In the plot, one can observe that in theory $\lambda(\omega)$ applied to the Sierpinski gasket produces a monotonically increasing speedup starting from $n \ge n_0 = 10$. The parallel space improvement in the number of threads used has also been included (dashed lines), showing a fixed exponential rate of improvement. 

For the intra-block mapping phase, the bounding sub-boxes approach was used. In order to know if a location is part of the fractal, each thread evaluates if $t_x \mathbin{\&} (b-1-t_y) == 0$ is true or false to know if it belongs to the Sierpinski gasket or not, respectively, with $\mathbin{\&}$ being the bitwise AND operator, $b$ the dimensional block size, and $t_x$ and $t_y$ the thread's coordinate in local space. Figure \ref{fig_thread_vs_block} illustrates how the block-space map and intra-block mapping are organized compared to a thread-space map. 
\begin{figure}[ht!]
\centering
\includegraphics[scale=0.05]{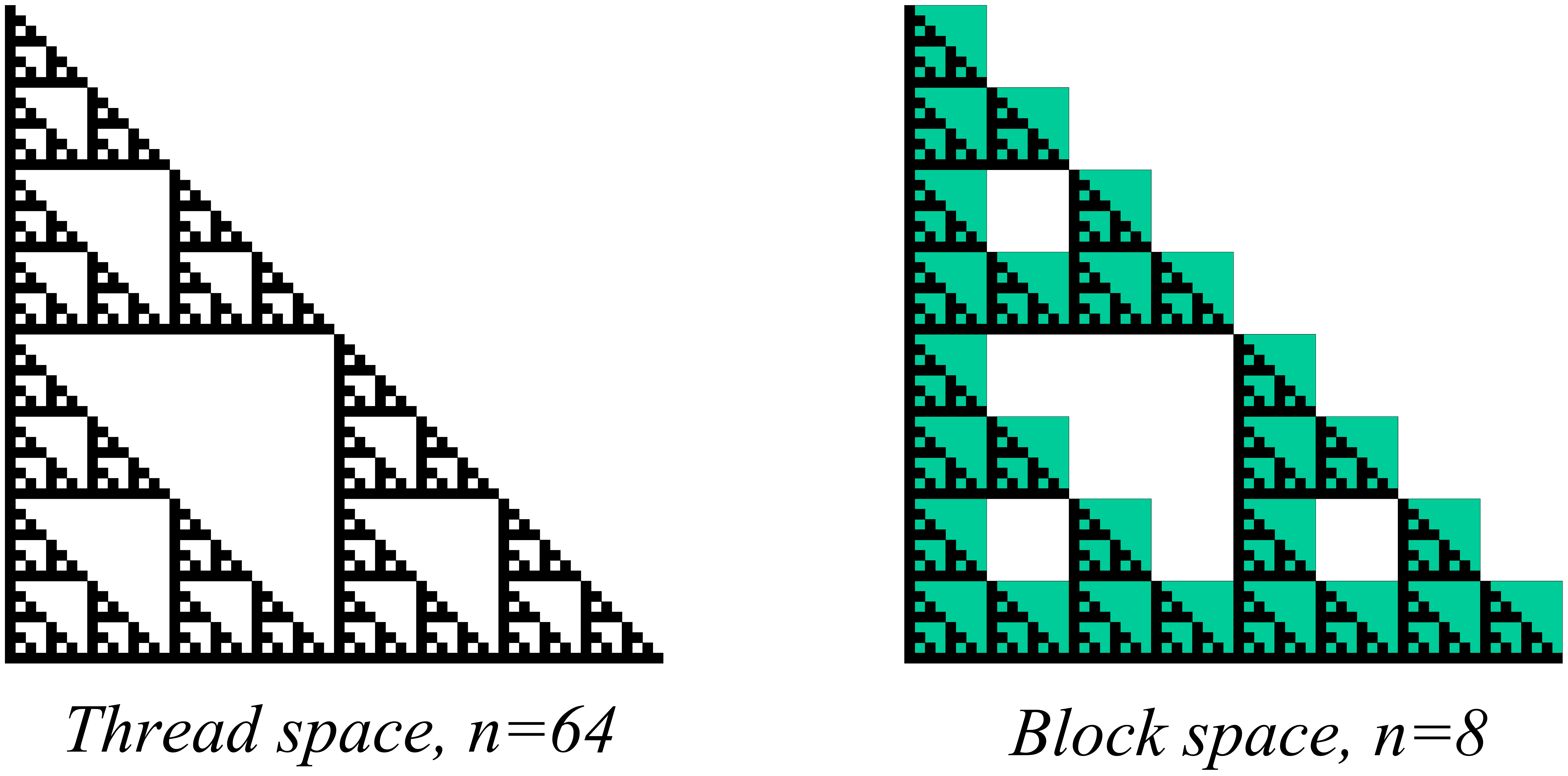}
    \caption{In thread-space mapping, threads are directly mapped one-to-one to the elements of the fractal of linear size $n=64$. In block-space mapping, $|B|
    = 8 \times 8$ and blocks of threads are mapped onto a simplified version (green) of the fractal of linear size $n_b=64/8=8$.}
\label{fig_thread_vs_block}
\end{figure}


\subsection{Implementation and Performance Results for the Sierpinski Gasket}
\label{sec_performance}
The case for the Sierpinski gasket was implemented using NVIDIA's CUDA \CC\ toolkit as a program that performs computations on the data elements of the fractal of side length $n$ (chosen at execution time) using both the bounding-box and $\lambda(\omega)$ approaches. In the case of $\lambda(\omega)$ the $x,y$ arithmetic reductions per-block coordinate $w$ from Eq. (\ref{eq:lambda-sierpinski}) are computed using the warp-shuffle parallel reduction, which allows efficient register-level communication among threads within a warp\footnote{A warp is a group of 32 threads that execute instructions in a lock-step mode and can also communicate their register data among themselves.}.
Experimental benchmarking of GPU thread maps is accompanied with work instructions in the GPU kernel to represent realistic application scenarios\footnote{Also it may not be clear if the compiler and scheduler optimizes the program, ignoring the mapping instructions, when no writes are performed on memory.}. The following three tests were designed, using different workloads:
\begin{itemize}
    \item Single write (SW): To write a constant value on all the elements of a Sierpinski gasket of scale level $r$, which is embedded in a $n\times n$ matrix initially filled with zeros. 
    \item Reduction (RD): To perform an arithmetic reduction with all the elements of the Sierpinski gasket.
    \item Cellular Automata (CA): To perform a Cellular Automaton simulation using a fractal adaptation of Conway's game of life-like rules. This adaptation still uses the Euclidean Moore neighborhood, but only considers as neighbors the cells that belong to the fractal and the cells of the empty embedded space are ignored in the neighborhood counting.
\end{itemize}

Different fractal sizes were tested in the range $r=0..16$ (up to $r = 15$ in tests RD and CA due to memory limitations), equivalent to embedding sizes of $n \times n = [\{1\times 1\}, ..., \{65536 \times 65536\}]$, and using different GPU block sizes in the range $\rho = 1,2,4,8,16,32$ in order to find the setting that provides the best performance for both the bounding-box and the $\lambda(\omega)$ approaches.
The average performance measures are taken by averaging 100 sub-averages, each one being an average time of 10 consecutive synchronized kernel calls. The standard error for each mean was below $1\%$. The hardware for performance test is listed in Table \ref{table_hardware}.
\begin{table}[ht!]
\normalsize
\caption{Hardware used for performance tests.}
\begin{center}
\begin{tabular}{|c|c|r|}
\hline
\# & Device	&	Model\\
\hline
 & GPU	&	Titan V, 5120 cuda cores 12GB \\
0 & CPU	&	Intel i7-6950X 10-core Broadwell \\
 & RAM	&	128GB DDR4 2400MHz\\
 \hline
 & GPU	&	Titan RTX, 4608 cuda cores, 24GB \\
1 & CPU	&	Intel i7-6950X 10-core Broadwell \\
 & RAM	&	128GB DDR4 2400MHz\\
\hline
\end{tabular}
\end{center}
\label{table_hardware}
\end{table}

Figure \ref{fig_performance} presents the speedup of $\lambda(\omega)$ over the bounding-box approach, as well as the running times for the two mapping techniques in all three different tests.
\begin{figure*}[ht!]
\includegraphics[scale=0.53]{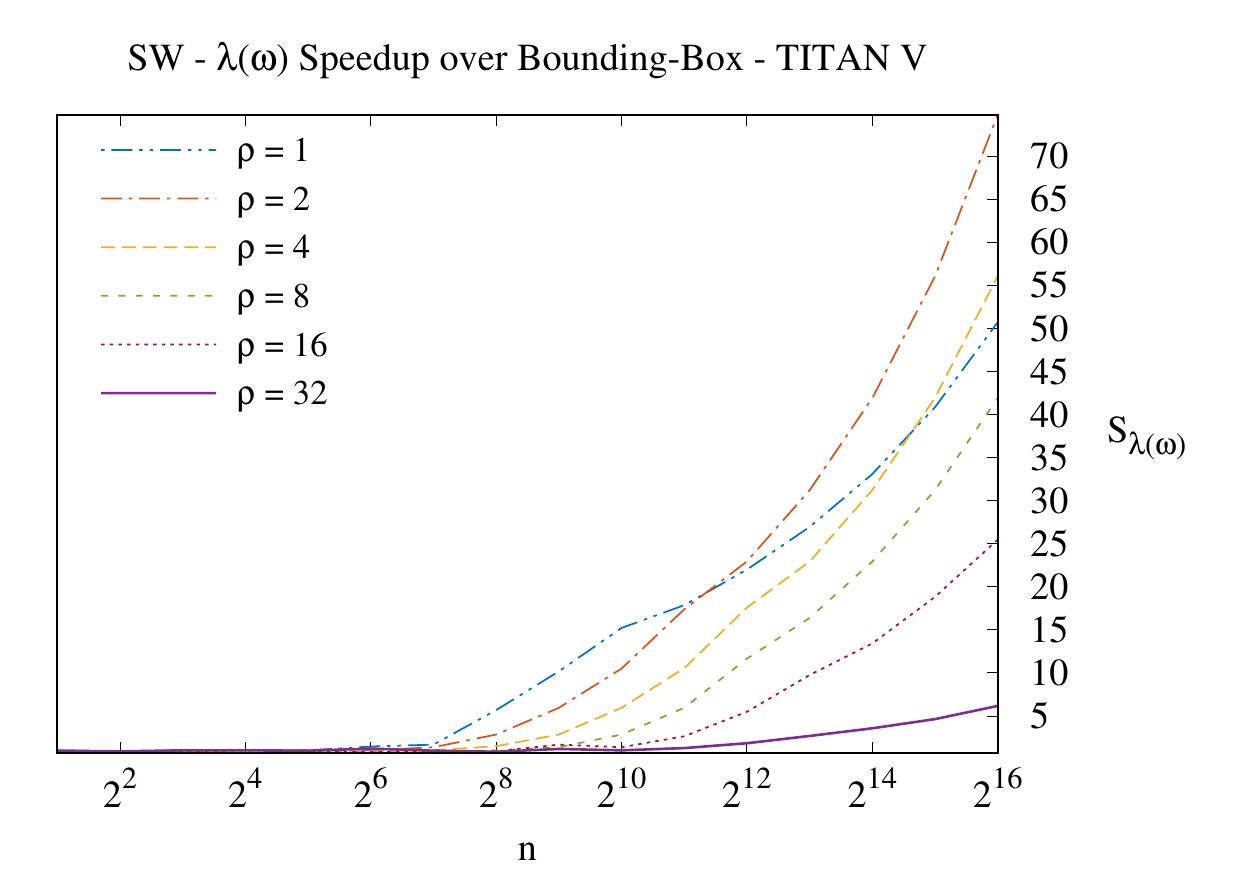}
\includegraphics[scale=0.53]{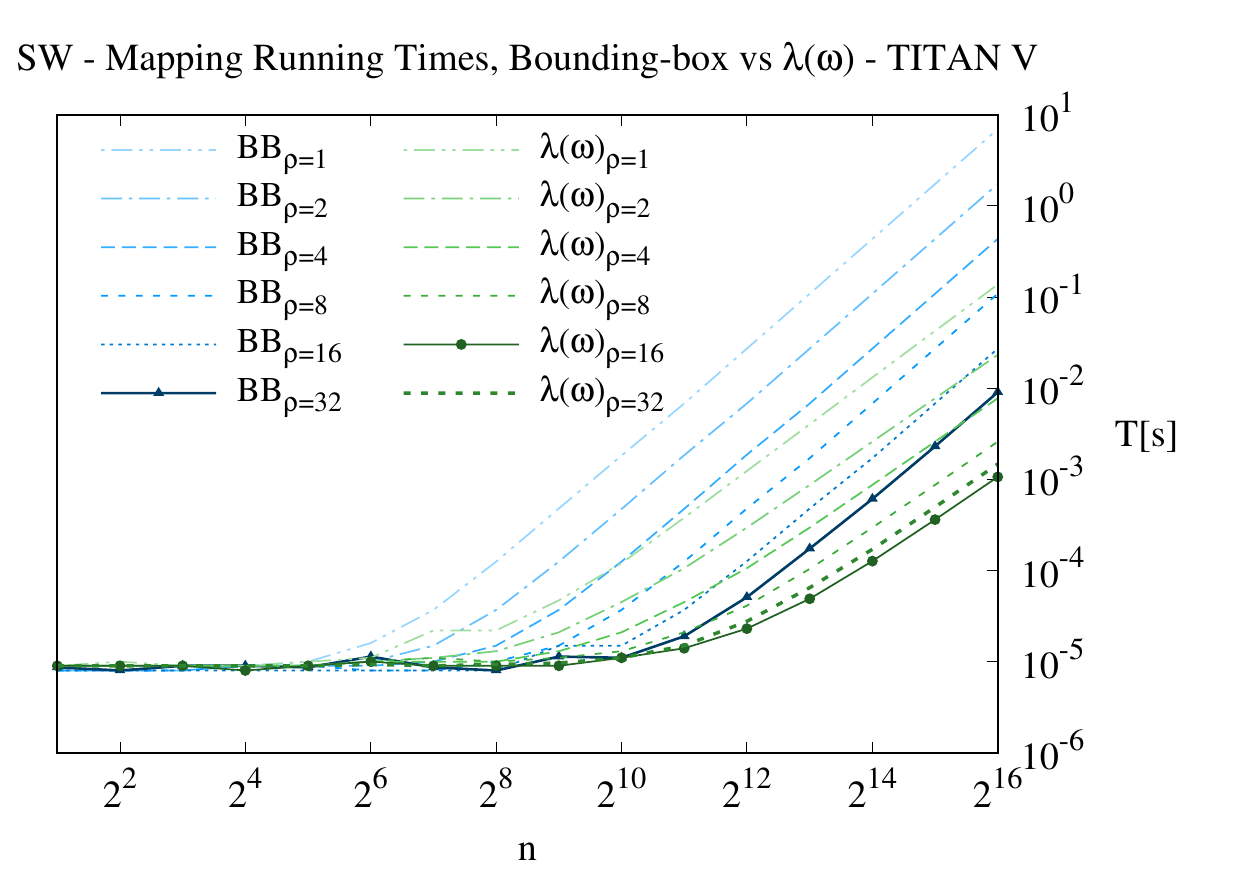}

\includegraphics[scale=0.53]{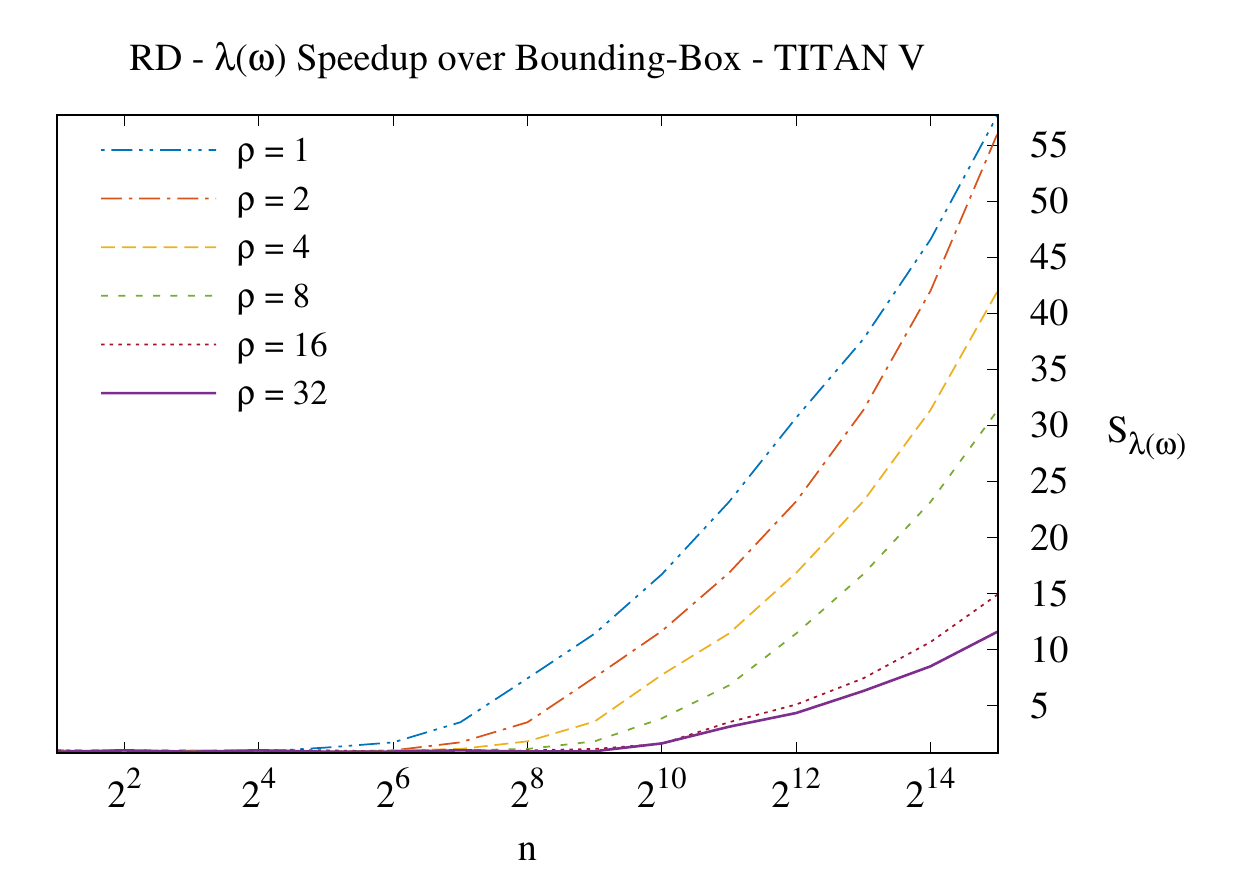}
\includegraphics[scale=0.53]{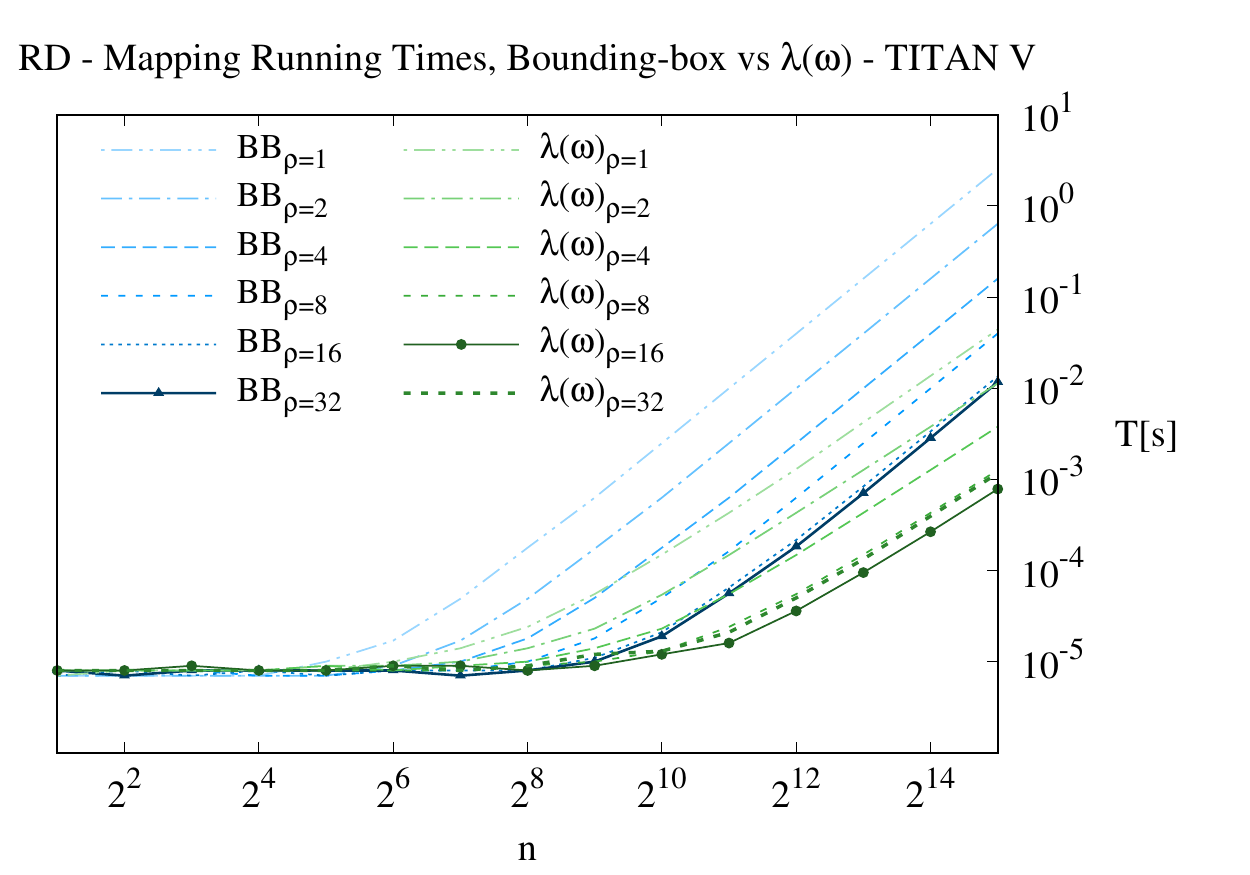}

\includegraphics[scale=0.53]{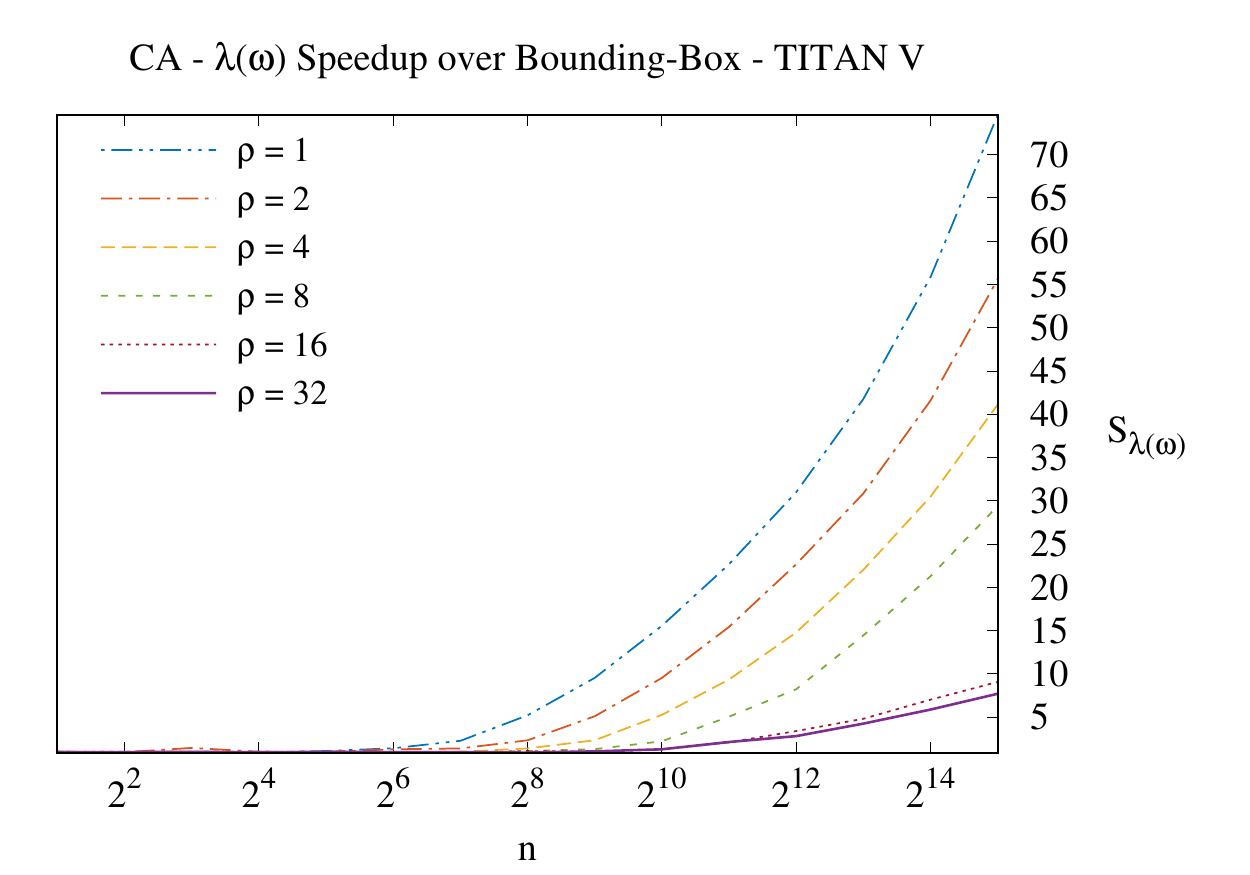}
\includegraphics[scale=0.53]{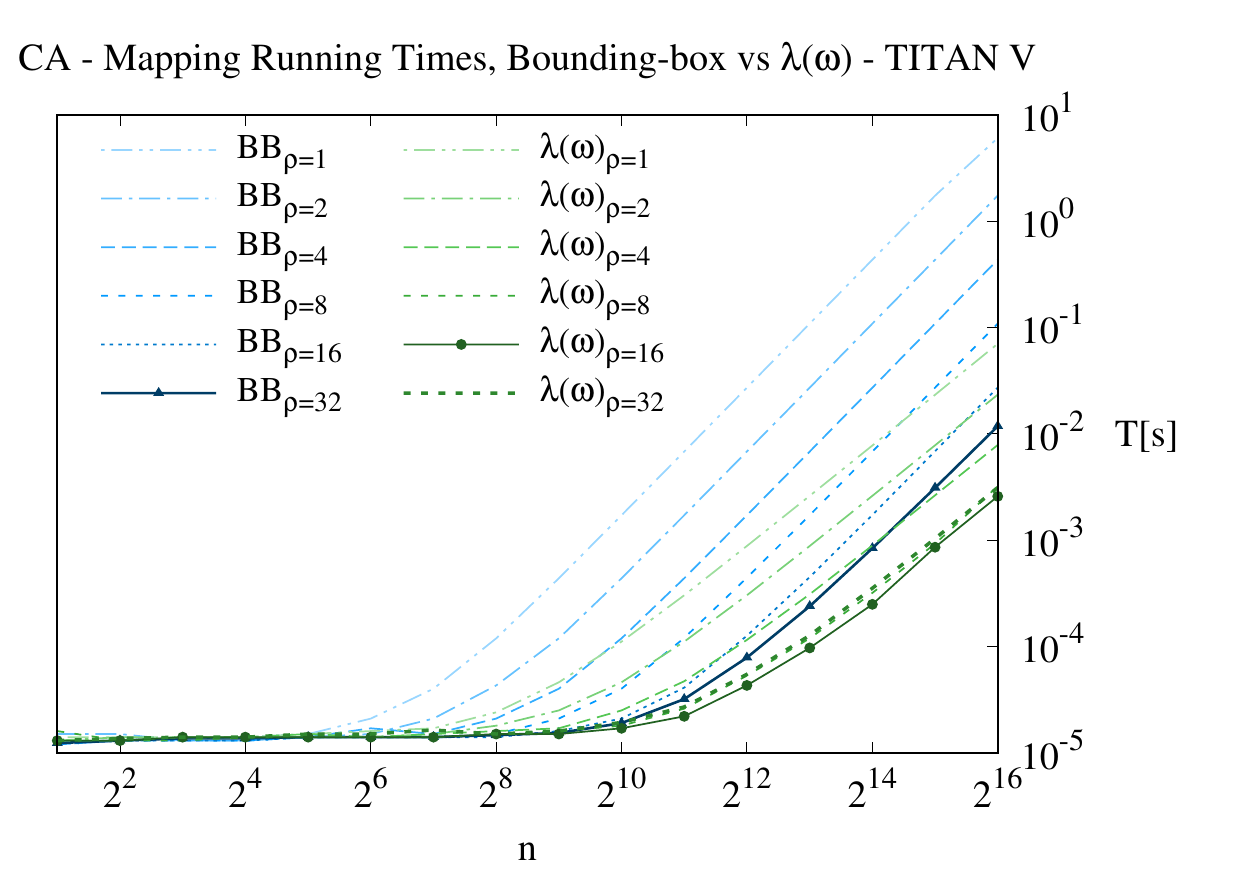}

\caption{The left column shows the speedup of $\lambda(\omega)$ with respect to the bounding-box approach at different block-size configurations and on the right column, 
their absolute running times at different block-size configurations. Each row shows the results of different test being: first row, test 1 simple write.
Second row, test 2 reduction. Third row, test 3 Cellular automata.}
\label{fig_performance}
\end{figure*}
For values of $n < 2^9$, one can note that only some curves offer speedup. Once $n > 2^9$, the speedup begins to increase for all block-size configurations,
reaching the higher values at $n=2^{16}=65536$, which was the highest problem size that fit in the GPU memory (\textit{i.e.} a fractal embedded in a region of $65536 \times 65536$). An important aspect to note from the speedup curves
is that for the largest possible block size, $|B|=\rho \times \rho=32 \times 32$, the $\lambda(\omega)$ map runs the tests between $6\times$ to $12\times$ faster than the bounding-box approach. Furthermore, as blocks become smaller in $\rho$, that improvement increases dramatically, reaching up to ~$75\times$ of speedup. 

The plot of the running times provides further insights on what configuration is the best suited for each mapping technique. By looking at the running times of
the small block configurations, one can note that regardless of their high speedup, their running times are the lowest, therefore these block sizes would not be used in practice. For the bounding-box approach
the best performance is obtained when the block-size is $|B| = 32 \times 32$. For $\lambda(\omega)$ the best performance is found when using a block of $|B| = 16\times 16$
threads. If the curves of the best configuration for each implementation are considered, \textit{i.e.}, the ones with bold mark from Figure \ref{fig_performance}, right, then the speedup provided by $\lambda(\omega)$ still reaches almost an order of magnitude. The running
time using other block sizes are still useful to visualize that as blocks become smaller, the value of $n_0$ where $\lambda(\omega)$ starts giving monotonically increasing speedup moves closer to the
origin, and vice versa. It is important to consider that the GPU, with its current organization and architecture, is not fully utilized when using very small block configurations, leading to an inferior performance than if larger blocks were used. Therefore, in practice large blocks would be utilized and by Theorem (\ref{theorem_speedup}), beyond $n = 2^{16}$ the speedup would keep increasing in favor of $\lambda(\omega)$.

We believe that these performance results can be useful for the GPU computing community as they show that for any modern programmable GPU, its performance can significantly improve when working with embedded NBB fractals just by employing a different thread map, not changing the rest of the application kernel code at all. In the next Section we describe how it is possible to further accelerate the performance of $\lambda(\omega)$ by adapting it to GPU 
tensor cores.

\subsection{Adapting $\lambda(\omega)$ for accelerated tensor core computation}
\label{sec_adaptingtc}
The Nvidia Volta GPU micro-architecture introduced a specialized hardware component called the Tensor Core. Actual GPUs of year 2018 and beyond can contain up to 640 tensor cores in addition to the regular GPU cores (which perform integer, floating point and read/write operations in parallel). Each tensor core is able to perform matrix-multiply-accumulate (MMA) operations on 4x4 matrices in one GPU clock cycle, which translates into a significant increase of TFLOPS compared to the classic operation mode of the GPU which is through the execution of floating point and integer arithmetic instructions. In order to make use of the Tensor Cores, the programmer must previously divide the problem into sub-problems of $16\times16$ sub-matrices, called fragments. This fragments are then given to the MMA subroutines which internally sub-divide them into 4x4 fragments to perform the operations in a warp-synchronized manner. 
Currently, as of 2020, details on how warps map to the fragments, or how tensor core perform the MMA operation are not fully specified by NVIDIA, moreover it is not guaranteed that a tensor-core based computation that works efficient in the Volta architecture (2017), will achieve the same level of performance in the Turing architecture (2019), or vice versa, as there are implementation details that are not exposed to the programmer. What is known is that the potential performance improvement will depend on how well the MMA operation can be exploited. These tensor-core related questions introduce additional motivations for knowing if the theoretical extra TFLOPS provided by a Tensor Core MMA operation, which were originally designed for Linear Algebra and Deep Learning, can be exploited to further speed up the calculation of $\lambda(\omega)$. 
Recent results support the idea that some computations may adapt well to tensor cores, such as the work of R. Carrasco \textit{et al.} where they study the potential speedup of computing the traditional arithmetic reduction based on tensor-core MMA operations \cite{8705253}. 
From their work, the authors conclude that the new tensor-core based reduction is in theory faster than a CUDA-Core based reduction. 
In this section we show how the computations for $\lambda(\omega)$ can be adapted as tensor-core MMA operations to calculate the tensor-core version of the map, namely $\lambda_{tc}(\omega)$.

In order speedup the calculation of $\lambda(\omega)$ with tensor cores, its equations must be encoded into a MMA operation in the form $D=A \times B+C$ where A, B, C, D are fragments, and C can be the same as D. There can be several ways to perform this encoding, and this section presents three variants of tensor core adaptation with their performance for the same tests.

\subsubsection{Variant 1: Simple per-block Tensor Core operation}
This variant employs one MMA computation for each block by exploiting the MMA-like behaviour from Eq. (\ref{eq:lambda-x}) and (\ref{eq:lambda-y}). This is done by expanding the sum, as well as the $\Delta_x^\mu, \Delta_y^\mu$ terms with the expression from Eq. (\ref{eq:deltas}), resulting into two sums of products; one to calculate $\lambda_x$ and the other to calculate $\lambda_y$. Each left multiplier from the sum terms is placed as an element of a row of fragment A and each right side of the sum terms is placed as a column of fragment B. Terms are placed in parallel and in the same order to match the corresponding pairements. 

The left side of both sums correspond to powers of two, from 0 to $\mu-1$, and are the same for $\lambda_x$ and $\lambda_y$. Therefore these factors can be encoded only using one row of fragment A and can be re-utilized for both $x$ and $y$ coordinates. The final encoding can be visualized in Figure \ref{fig_A_B_arangement}.

\begin{figure}[ht!]
    \centering
        $ A=
        \begin{pmatrix}
        2^0 & 2^1 & \dots & 2^{\mu-1}\\
        0 & 0 & \dots & 0\\
        \vdots & \vdots & \ddots & \vdots\\
        0 & 0 & \dots & 0\\
        \end{pmatrix}
        $
        $B=
        \begin{pmatrix}
        \tau^{1}_x & \tau^{1}_y & 0 & \dots & 0\\
        \tau^{2}_x & \tau^{2}_y & 0 & \dots & 0\\
        \vdots & \vdots & \vdots & \ddots & \vdots\\
        \tau^{\mu}_x & \tau^{\mu}_y & 0 & \dots & 0\\
        \end{pmatrix}
$
        \caption{Encoding of Variant 1 in a MMA manner. Note that this matrices are dimensions $\mu x \mu$ and when loaded into fragments of $16\times16$, remaining elements are filled with zeroes.}
    \label{fig_A_B_arangement}
\end{figure}

An important technical note is that Fragment B was defined as column major matrix and fragment A as a row major to ease the memory access during the calculations and minimize data divergence. Once the Tensor core operation is done, the results $\lambda_x$ and $\lambda_y$ become the first and second elements of the first row of fragment D, respectively.

\subsubsection{Variant 2: Sub-blocked tensor core operation}
This variant shares the principle of Variant 1, but expands the idea it by subdividing the block of threads into sub-blocks, calculating more block coordinates (one for each sub-block) still using one tensor core MMA operation. Since every fragment has 16 rows and columns, there is the potential of calculating 16 different values; 8 pairs of $(\lambda_x, \lambda_y)$.  
The approach assumes a thread block size large  enough, to contain 4 sub-blocks of size  $b/2 \times b/2$ threads that are still large enough to produce an efficient computation. These sub-blocks are now treated as independent blocks of threads that are not yet mapped and sit in parallel space ready to be mapped with $\lambda(\omega)$. Figure \ref{fig_grouping} shows an example using a block size of $32\times32$ subdividing into sub-blocks of $16 \times 16$.

\begin{figure}[ht!]
\centering
\includegraphics[scale=0.55]{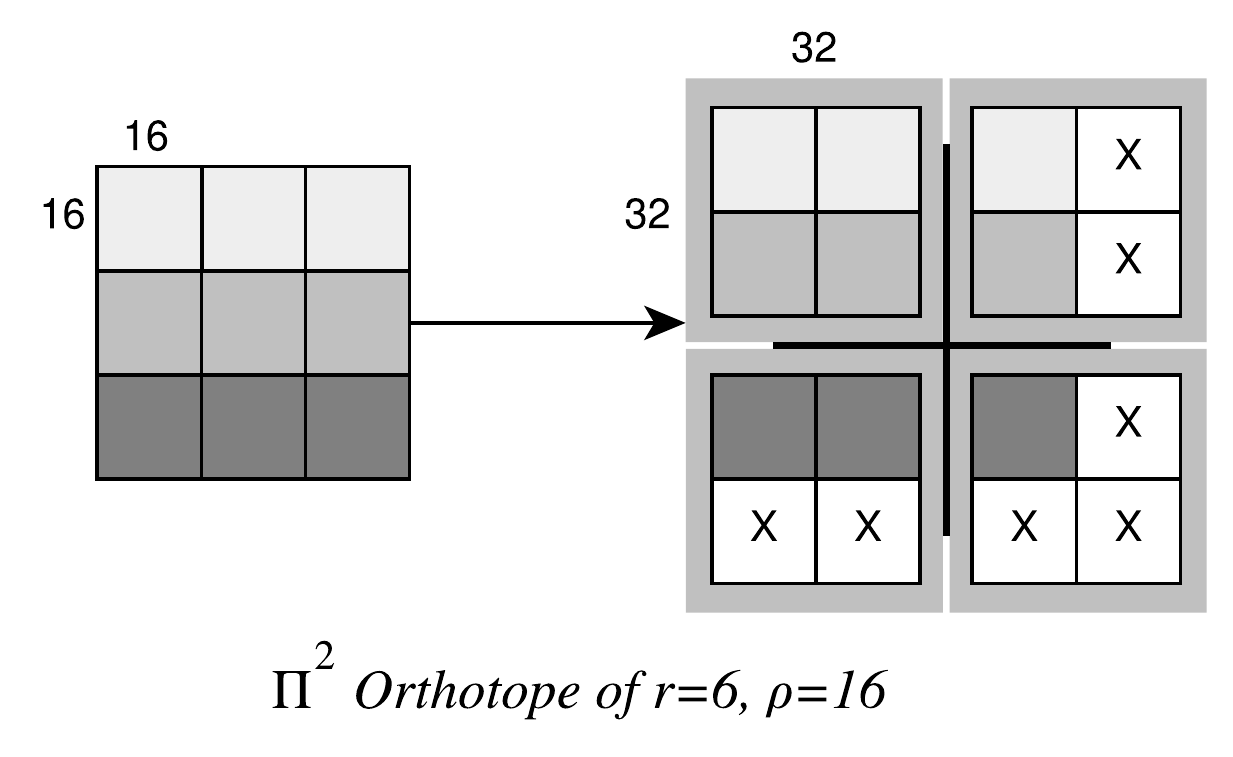}
    \caption{On the left a 2-Orthotope for a Sierpinski gasket of scale level is 6 with a sub-block size of $16 \times 16$. On the right, the 2-Orthotope when treated with block size of $32\times32$ and then subdivided.}
\label{fig_grouping}
\end{figure}

Once the MMA operation is done, the mapped coordinates for each sub-block are found in the first row of the resulting fragment $D$.
It is worth noticing that this approach introduces some chunks of unused threads which lie outside of the 2-Orthotope generated. The uncolored blocks marked with X in the right side of Figure \ref{fig_grouping} represents unused sub-blocks of threads. These extra threads do not introduce a significant cost as they are upper bounded by $O(\sqrt{n^{\mathcal{H}}})$ (\textit{i.e}, the perimeter of the packed fractal) in comparison to the domain of the 2-Orthotope which is $O(n^{\mathcal{H}})$.

\subsubsection{Variant 3: Full A, B and C usage}
The two variants described rely only on fragments A and B for its calculation while most of their fields are empty, while fragment C is unused, therefore not taking advantage of the addition operator of the MMA operation. This third variant maintains the same encoding for A and B, but including C in the calculation, filling completely all fields of the 3 fragments. Figure \ref{fig_A_B_arangementC} shows the arrangement of each matrix.

\begin{figure}[ht!]
    \centering
        $ A=
        \begin{pmatrix}
        2^0 & 2^1 & \dots & 2^{\mu-1}\\
        2^0 & 2^1 & \dots & 2^{\mu-1}\\
        \vdots & \vdots & \ddots & \vdots\\
         2^0 & 2^1 & \dots & 2^{\mu-1}\\
        \end{pmatrix} 
        $
        $B_x=
        \begin{pmatrix}
        \tau^{1}_x & \tau^{1}_x & \dots & \tau^{1}_x\\
        \tau^{2}_x & \tau^{2}_x & \dots & \tau^{2}_x\\
        \vdots      & \vdots    & \ddots & \vdots\\
        \tau^{\mu}_x & \tau^{\mu}_x & \dots & \tau^{\mu}_x\\
        \end{pmatrix}
$
$$
$$
$C_x=
        \begin{pmatrix}
        t_{1,1}^x & t_{1,2}^x & \dots & t_{1,\mu}^x\\
        t_{2,1}^x & t_{2,2}^x & \dots & t_{2,\mu}^x\\
        \vdots & \vdots & \ddots & \vdots\\
        t_{\mu,1}^x & t_{\mu,2}^x & \dots & t_{\mu,\mu}^x\\
        \end{pmatrix}
$
        \caption{MMA scheme of variant 3. Fragments $B_y$ and $C_y$ share the same representation as their $x$ counterpart. Note that these matrices are dimensions $\mu x \mu$ and when loaded into fragments, the extra elements get filled with zeroes.}
    \label{fig_A_B_arangementC}
\end{figure}
The result of the tensor Core MMA is now a thread coordinate in data space for each thread, whilst in the previous variants, it was a block coordinate in data space for all threads within that block.
This means that every thread accesses its corresponding coordinate in fragment D using its parallel space coordinates in a 1:1 mapping.
It is implied that this variant necessarily also encodes the intra-block mapping phase into the tensor core operation. Bounding sub-boxes was the mapping used by default.
In order for this Variant to work, 2 tensor core MMA operations must be performed per block, one for $\lambda_x$ coordinate  and the other for $\lambda_y$ of each thread.
This variant was developed for block size $\rho=16$ to fit the 256 $\lambda$ coordinates of the 256 threads. The $\rho=32$ version is also possible and would require to sub-divide the block into four sub regions of $16\times 16$ performing 8 tensor core operations in total, with the access to data space being contiguous by the sub-regions for they are not treated as independent blocks like in variant 2.

\subsubsection{$\lambda_{tc}(\omega)$ results}
The tensor core-based methods were tested with the same workloads and hardware as $\lambda(\omega)$ and were compared within their 
corresponding versions. The speedups of the tensor core variants with respect to the regular $\lambda(\omega)$ are shown in figure \ref{fig_tensorSpeedup}. It is important to review the results of both Volta and Turing architecture, considering that the internal implementation of tensor core operations may give different performances.

\begin{figure*}[ht!]

\includegraphics[scale=0.53]{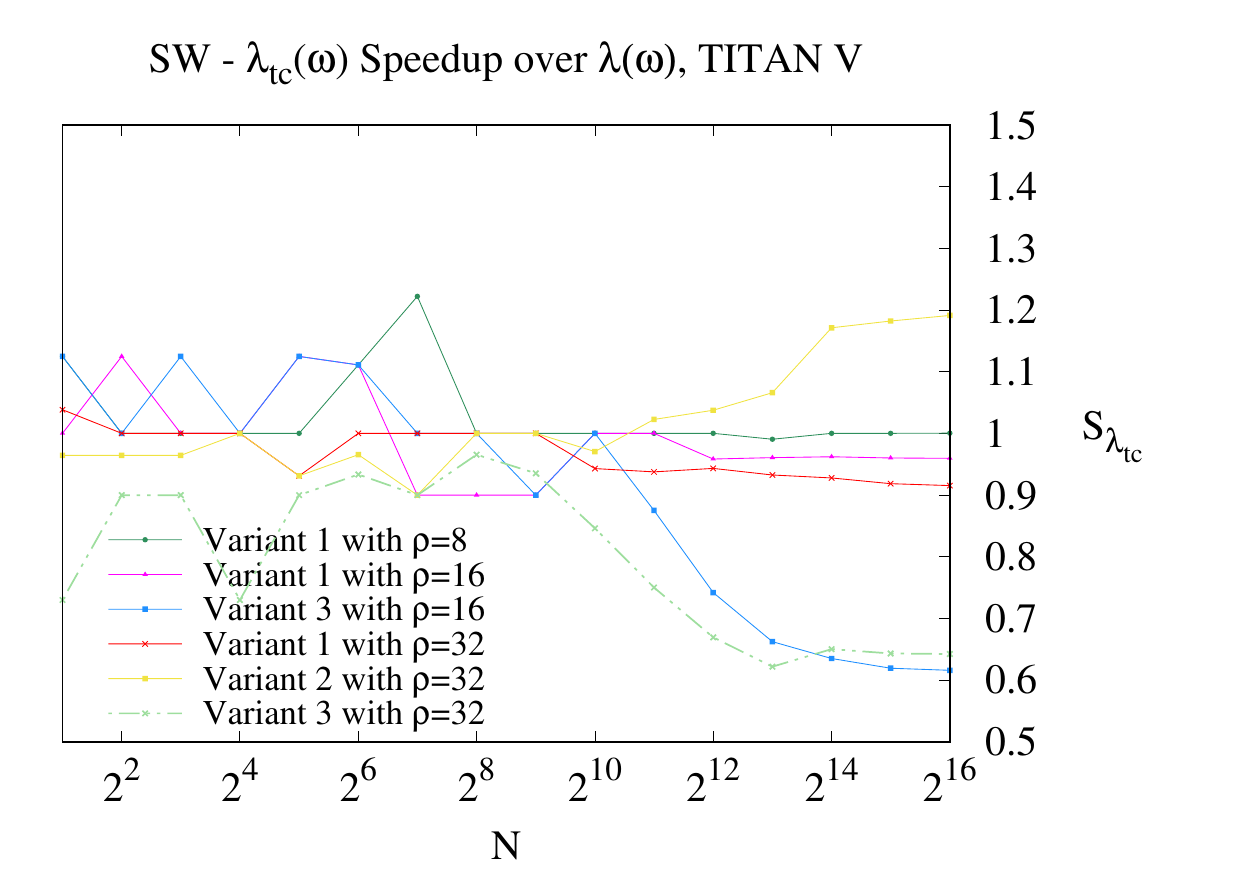}
\includegraphics[scale=0.53]{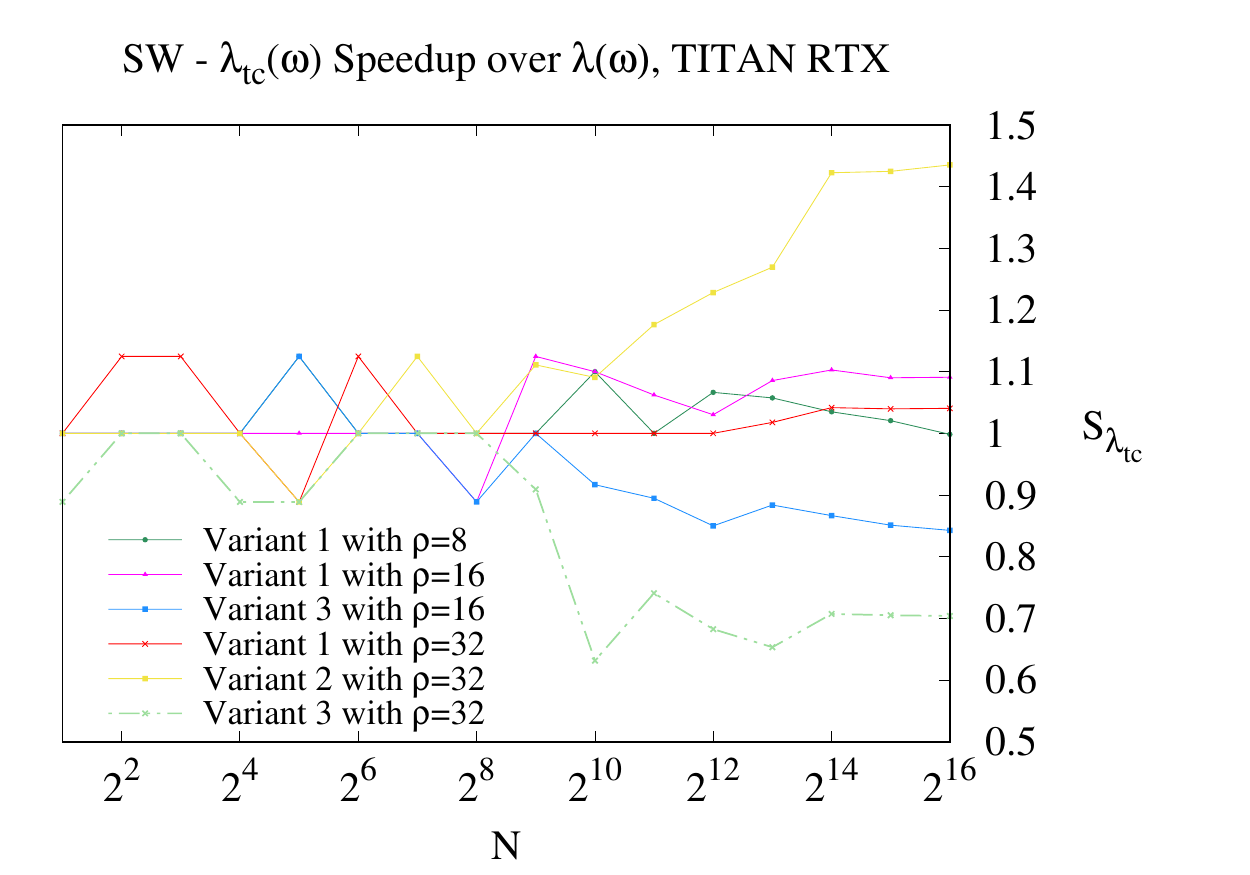}

\includegraphics[scale=0.53]{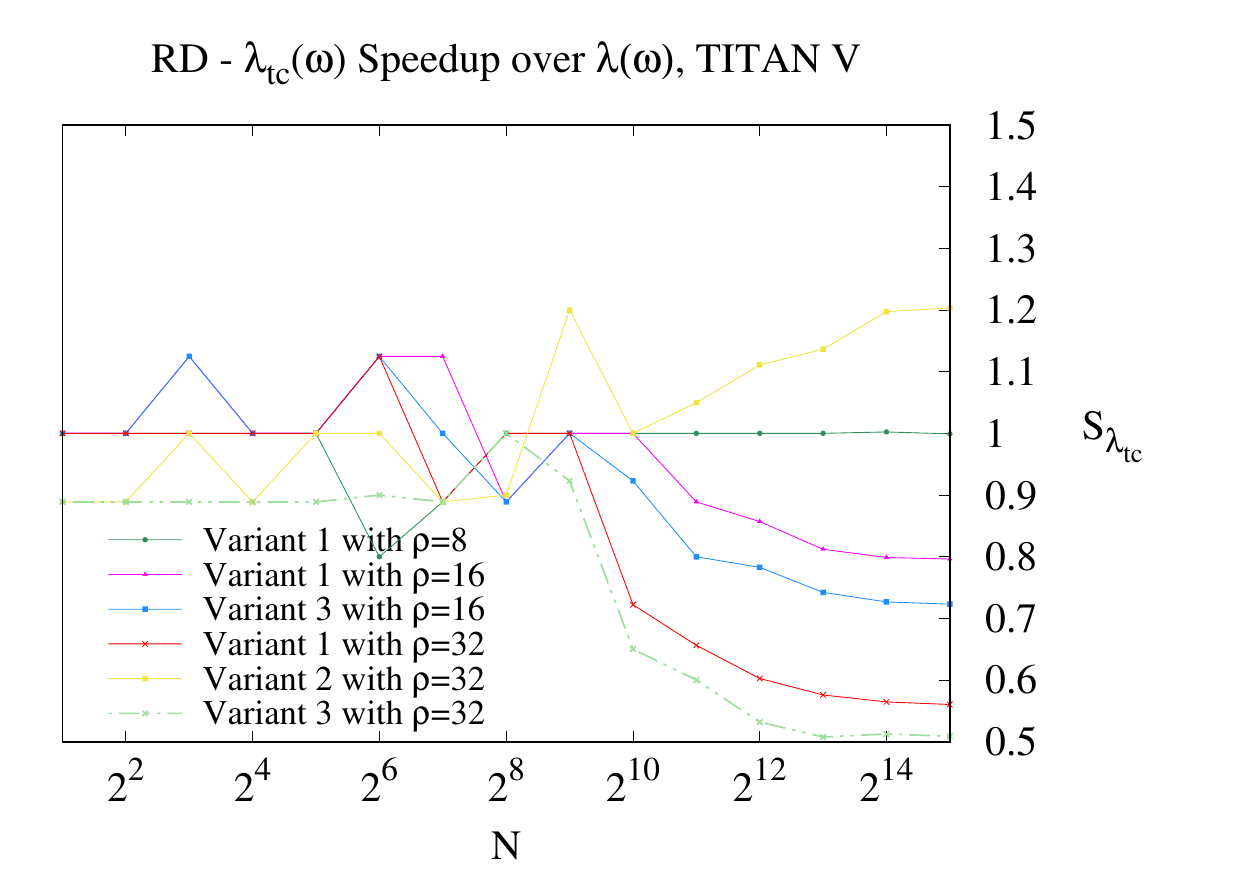}
\includegraphics[scale=0.53]{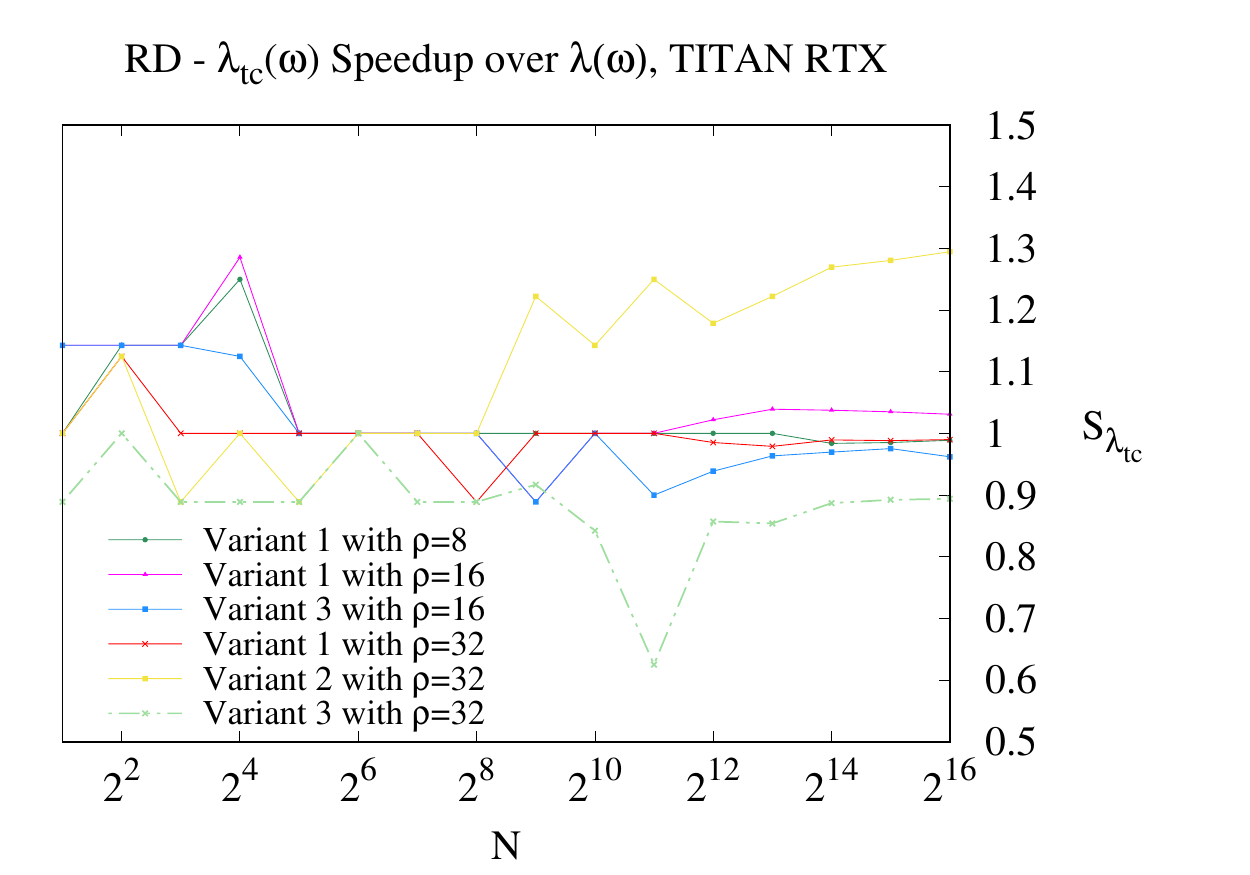}

\includegraphics[scale=0.53]{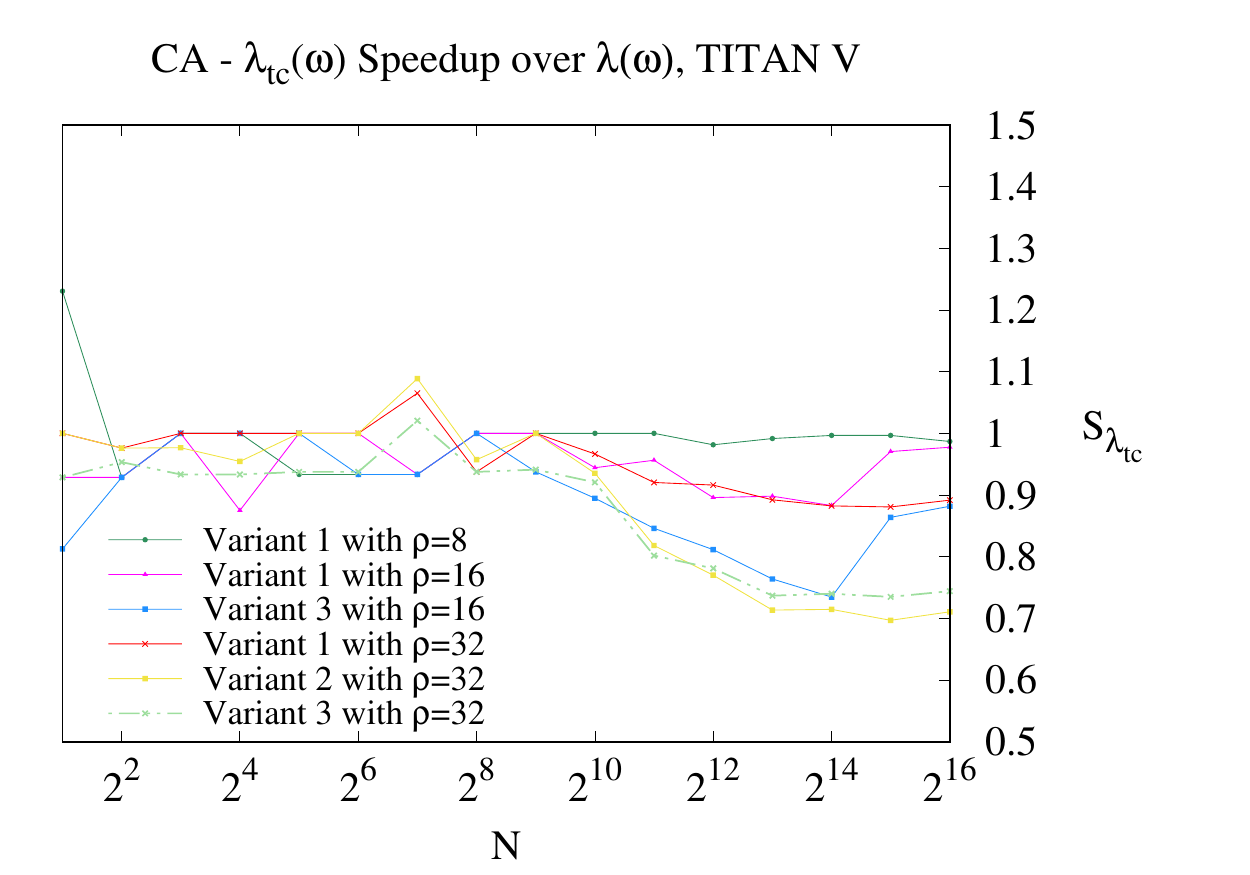}
\includegraphics[scale=0.53]{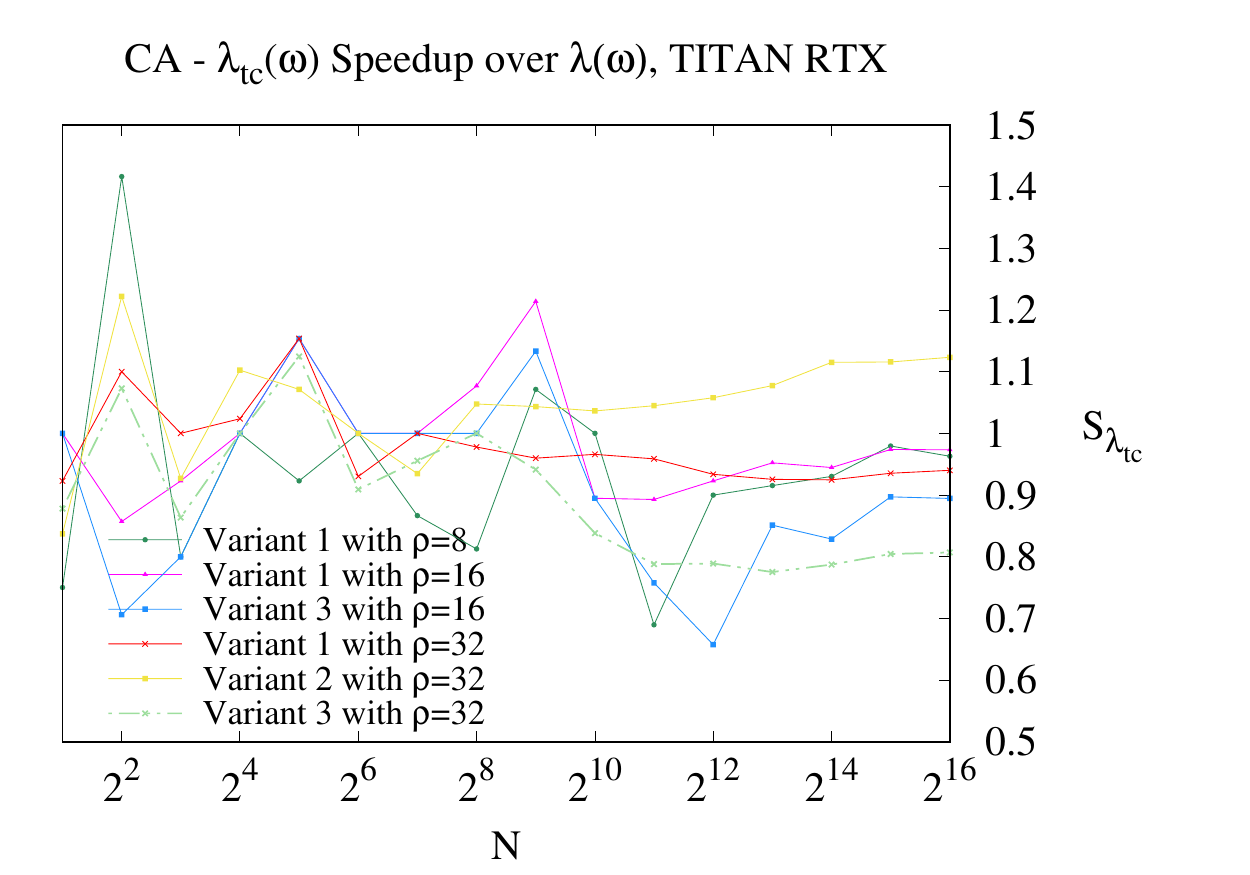}

\caption{The graphs shows the speedup of $\lambda_{tc}(\omega)$ with respect to non-tensor core based $\lambda(\omega)$. The left column shows the results with a TITAN V GPU (Volta architecture) and right column results with a TITAN RTX GPU (Turing architecture). 
}
\label{fig_tensorSpeedup}
\end{figure*}

Starting with the first test, the single-write (SW), results with the TITAN V show that when $n \ge n_0 = 2^{10}$, speedup curves reach a stable behavior, whereas with $n < n_0$ the speedup curves show unstable behavior oscillating around $1.0$ of speedup. In the large scale regime variant 2 gives up to $~20\%$ of extra performance over the regular $\lambda(\omega)$ map. 
With the TITAN RTX, the speedup curves show that again variant 2 increases $\lambda(\omega)$ performance up to a $~40\%$, and variant 1 becomes a usable choice as it gets a speedup above 1.0 in all configurations. In both GPUs, variant 3 results in poor performance when $n$ is greater than $n_0$ with both $\rho=16$ and $\rho=32$ block sizes. In the reduction (RD) test the results were no different from SW, with variant 2 getting the best performance with a $~20\%$ and $~30\%$ of increased performance in TITAN V and RTX respectively. Variant 1 is only beneficial in TITAN RTX with $\rho=16$, and variant 3 is inefficient in all block configurations.
Lastly, results on the Celullar Automata test (CA) when $n<n_0$ shows the same behaviour as previous tests, but as $n$ grows from $n_0$ in TITAN V, speedup of all curves start to decrease below 1, the same happens with TITAN RTX except for variant 2 that gets a positive performance boost. The general performance hit observed for all variants may be a consequence of the computation patterns found in the CA simulation.

Summarizing the tensor core results, the variant that showed better results overall was variant 2 with up to a $~40\%$ of performance boost over non-tensor core lambda when $n>2^{10}$. Variant 1 is also better under certain configurations with a $5\sim 10\%$ of performance boost. Variant 3 showed slower performance than the regular $\lambda(\omega)$ thus can be discarded. 

\section{Discussion and Conclusions}
\label{sec:discussion-conclusions}
This work has shown that the $\lambda(\omega)$ map proposed for NBB fractals leads to a significant performance speedup both in theory and in experimental tests with an embedded Sierpinski gasket. 
The analysis and formulation of $\lambda(\omega)$ has provided three important 
results in the theoretical aspect; (1) There exists a correspondence between a quasi-regular $2$-orthotope and NBB fractals, (2) such correspondence can be computed in just $\mathcal{O}(\log_2 \log_2(n))$ time and (3) the total work required for mapping the $2$-orthotope used with $\lambda(\omega)$ is asymptotically smaller than the work generated by the bounding box approach, leading to a monotonically increasing speedup starting from $n \ge n_0$. In particular, Theorem \ref{theorem_speedup} serves as a guarantee that using $\lambda(\omega)$ on any NBB fractal will provide a significant performance speedup starting from a certain fractal scale onward, as the speedup monotonically increases with $n$.

The experimental performance results confirm the theoretical results, showing monotonically increasing speedup once $n \ge n_0 = 2^9$ and up to $9\times$ of speedup over a bounding-box approach using optimal block-size settings for each approach. Using smaller block-sizes leads to even higher speedups (up to $75\times$) but slower running times. It is uncommon for GPU applications to use small block sizes, but in case it is required, a significant performance improvement is available with this approach. Still, the exploration of GPU performance under different block-sizes has allowed to understand that small block sizes behave as the theoretical results in a strong way, while the largest block sizes, although still produce monotonically increasing speedup, behave as the theory in a weaker form.
The adaptation of $\lambda(\omega)$ to use tensor core computation offered up to $\sim40\%$ of extra performance. In order to achieve efficient tensor core computation it is important to exploit communication between tensor core fragments and shared memory as most as possible, and in the case of non-Machine Learning tasks, to codify the computation with the least redundant data into the fragments. The closer the task is to linear algebra, the easier this adaptation will be.

It is important to note that the GPU thread map presented in this work, along with the tensor core optimization, can be implemented for any fractal belonging to the NBB family. The implementation 
of the case study from this work, including the three tests with and without tensor-core adaptations, is available for the community at \url{https://github.com/crinavar/xxxxx}\footnote{Note to the reviewers: the repository will be made open to the community in the published version of this article.}. Future work on this line can follow two paths; (1) further study GPU thread mapping on fractals and extend the NBB family to include other fractals which use rotations in the replicas, such as the Koch curve, and (2) evaluate the performance of compact fractal manipulation in GPU, that is, to pack the data space into an orthotope as the parallel space, and allow operations on the structure without decompressing the fractal, just by using $\lambda(\omega)$ and $\lambda(\omega)^{-1}$ for unrolling and rolling the computations. This last path could allow handling much larger fractals in GPU as the memory used would be in the order of $n^{\mathcal{H}}$. Future research in these directions can provide important insights on the potential benefits of efficient GPU computing for fractal geometry.


\section*{Acknowledgment}
This work was supported by the research projects FONDECYT N$^o$ 11180881 and 1181506, both from CONICYT, as well as by the Nvidia CUDA Research Center at the Department of Computer Science (DCC) from University of Chile and the Millenium Institute Foundational Research on Data (IMFD).





\bibliographystyle{elsarticle-num}
\bibliography{main}







\end{document}